\documentclass[reqno,11pt]{amsart}
\usepackage{cite}
\usepackage{graphicx}
\usepackage{subfigure}
\usepackage{amssymb,amsmath,amsthm,amscd}
\usepackage{mathrsfs}
\usepackage{tipa}
\usepackage{epic}
\usepackage{fullpage}
\usepackage{amsfonts}
\usepackage{arydshln}
\usepackage{latexsym}
\usepackage{bm}
\graphicspath{{figures/}}
\usepackage{color}

\DeclareMathOperator{\sech}{sech}
\DeclareMathOperator{\csch}{csch}

\newtheorem{Thm}{Theorem}

\numberwithin{equation}{section}
\newcommand{\wt}{\widetilde}
\newcommand{\wh}{\widehat}

\newcommand{\st}{\hbox{\tiny\it{T}}}
\newcommand{\nn}{\nonumber}

\begin{document}
\title{Solutions to nonlocal nonisospectral (2+1)-dimensional breaking soliton equations}
\author{Hai-jing Xu, Wei Feng, Song-lin Zhao$^{*}$\\
\\\lowercase{\scshape{
Department of Applied Mathematics, Zhejiang University of Technology,
Hangzhou 310023, P.R. China}}}
\email{*Corresponding Author: songlinzhao2021@126.com}

\begin{abstract}

Nonlocal reductions of a nonisospectral (2+1)-dimensional breaking soliton
Ablowitz-Kaup-Newell-Segur equation are discussed on the base of double Wronskian reduction technique.
Various types of solutions, including soliton solutions and Jordan-block solutions, for the resulting
nonlocal equations are derived. Dynamics of these obtained solutions are analyzed and illustrated.

\end{abstract}
	
	
\maketitle
\section{Introduction}
\label{sec-1}

Among the nonlinear partial differential equations, of particular
interest are the nonlocal integrable systems that admit ``parity-time symmetry''. The activities
in the field of nonlocal integrable system were initiated by Ablowitz and Musslimani \cite{AbMu-2013},
who introduced an integrable nonlocal nonlinear Schr\"{o}dinger (NLS) equation
\begin{align}
\label{n-NLS}
iu_t(x,t)+u_{xx}(x,t)+u^2(x,t)u^*(-x,t)=0,
\end{align}
as a new reduction of the Ablowitz-Kaup-Newell-Segur (AKNS) hierarchy. This equation
is parity-time symmetric because it is invariant under the action of the parity-time operator,
i.e., the joint transformations $x\rightarrow -x,~t\rightarrow -t$ and complex conjugation $^*$.
Since then, the nonlocal integrable systems have attracted much attention from both mathematics and the physical application
of nonlinear optics and magnetics \cite{Gad,Mak,Muss}. Up to now, more and more nonlocal integrable equations have been established
\cite{AbMu-SAPM,Lou,Pekcan}, including
nonlocal versions of the Korteweg-de Vries, the modified Korteweg-de Vries, the sine-Gordon,
the nonlinear ``loop soliton'', the Davey-Stewartson equations, etc.
Many traditional methods, such as the inverse scattering transformation,
the Riemann-Hilbert approach, the Hirota's bilinear method, the Darboux transformation and the Cauchy matrix approach,
have been used to search for exact solutions to the nonlocal integrable systems
\cite{AM-Nonl-2016,Yan-AML-2015,Zhou-ZX-1,XU-AML-2016,ChenZ-AML-2017,CDLZ,FZS-IJMPB,FZ-ROMP}.

Very recently, there have been some research on the nonlocal nonisospectral integrable equations.
With the help of the reduction technique developed in \cite{CDLZ}, by using a second-order nonisospectral AKNS system,
Liu, Wu and Zhang \cite{LWZ} investigated a nonlocal
Gross-Pitaevskii equation with a parabolic potential and a gain term,
and constructed its one-soliton solution, two-soliton solutions and
Jordan-block solutions. Some interesting dynamics were revealed.
Subsequently, the authors of the present paper \cite{FZ-IJMPB,Xu-Zhao} also used this approach to study the nonlocal reduction of three
nonisospectral AKNS type equations, including the second order nonisospectral AKNS equation,
the third order nonisospectral AKNS equation and the first negative order nonisospectral AKNS equation.
Consequently, some real and complex nonlocal nonisospectral NLS,
modified Korteweg-de Vries and sine-Gordon equations were presented.

As a typical (2+1)-dimensional generalization of the NLS equation,
the (2+1)-dimensional breaking soliton equation
is usually used to describe the (2+1)-dimensional
interaction of a Riemann wave propagating along the $y$-axis with a long-wave propagating along the $x$-axis \cite{Bogo-1990}.
In recent years, the research of nonlocal (2+1)-dimensional breaking soliton type equations has received remarkable progress.
Zhu and Zuo \cite{Zuo} proposed some nonlocal (2+1)-dimensional breaking soliton equations and showed that most of them could be decomposed into two (1+1)-dimensional
integrable systems. Moreover, Darboux transformations and some exact solutions were constructed (see also \cite{Zhu}).
Wang, Wu and Zhang \cite{WWZ-CTP} presented double Wronskian solutions to the nonlocal isospectral (2+1)-dimensional breaking soliton AKNS hierarchy
and illustrated the dynamical behaviors of some obtained solutions.

In this paper, we are interested in the nonlocal reductions of the following nonisospectral (2+1)-dimensional breaking soliton AKNS equation \cite{Yao}
\begin{subequations}
\label{n-(2+1)-DBS}
\begin{align}
& u_t=-y(u_{xy}-2u\partial^{-1}_{x}(uv)_y)-\frac{x}{2}(u_{xx}-2u^2v)-u_{x}+u\partial^{-1}_{x}(uv), \\
& v_t=y(v_{xy}-2v\partial^{-1}_{x}(uv)_y)+\frac{x}{2}(v_{xx}-2uv^2)+v_{x}-v\partial^{-1}_{x}(uv),
\end{align}
\end{subequations}
where $u$ and $v$ are two functions of $x,~y,~t$ and
$\partial_x^{-1}=\frac{1}{2}\left(\int_{-\infty}^x-\int^{+\infty}_x\right)\cdot dx$.
When $y=x$, system \eqref{n-(2+1)-DBS} is nothing but the second order nonisospectral AKNS equation \cite{Chen-book}.
We plan to impose real nonlocal reduction, respectively, complex nonlocal reduction on the
system \eqref{n-(2+1)-DBS} and get the real and complex nonlocal nonisospectral (2+1)-dimensional breaking soliton NLS equations.
We will employ the reduction approach \cite{ChenZ-AML-2017,CDLZ} to get solutions of the resulting equations.
For simplicity, we call equation \eqref{n-(2+1)-DBS} NBS-AKNS for short. Analogously,
we denote the real nonlocal resulting equation by RNNBS-NLS,
respectively, the complex nonlocal resulting equation by CNNBS-NLS.

The paper is organized as follows. In Sec. 2, we briefly recall Lax representation and double Wronskian solutions of
the NBS-AKNS equation \eqref{n-(2+1)-DBS}. In Sec. 3 and  Sec. 4,
we, respectively, study real nonlocal reduction and complex nonlocal reduction
of the system \eqref{n-(2+1)-DBS}. By solving determining equations, we give
soliton solutions and Jordan-block solutions for the resulting nonlocal equations.
Meanwhile, we illustrate and analyze the dynamics of some obtained
solutions. Section 5 is devoted to the conclusions.


\section{Lax representation and double Wronskian solutions of \eqref{n-(2+1)-DBS}}\label{sec-2}

In this section, we briefly recall the Lax representation and double Wronskian solutions of the system \eqref{n-(2+1)-DBS}.
For the details one can refer to Ref. \cite{Yao}.

\subsection{Lax representation}

System \eqref{n-(2+1)-DBS} is related to the famous AKNS spectral problem \cite{AKNS-Lax}
\begin{subequations}
\label{n-(2+1)-DBS-Lp}
\begin{align}
\label{AKNS-sp}
&\Phi_x=M\Phi,\quad
M=\left(\begin{array}{cc}
-\lambda & u \\
v & \lambda
\end{array}\right),
\end{align}
together with time evolution
\begin{align}
\label{AKNS-time}
\Phi_t=2y\lambda \Phi _y+N \Phi, \quad
N=\left( \begin{array}{cc}
A & B \\
C & -A
\end{array}
\right),
\end{align}
\end{subequations}
where $\Phi=(\Phi_{1}, \Phi_{2})^{\st}$ and $\lambda$ is the spectral parameter. Here and hereafter,
$K^{\st}$ means transpose of matrix $K$.

The nonisospectral zero curvature equation $M_{t}-N_{x}+[M,N]-2y\lambda M_y=0$ gives rise to
\begin{subequations}
\label{u-v-A}
\begin{align}
& A=\partial^{-1}(v,u)\left(\begin{array}{c}
-B \\ C
\end{array}\right)+(\lambda_t-2\lambda \lambda_y)x+A_0,\\
& \label{q-r} \left(\begin{array}{c}
u \\ v
\end{array}\right)_t=L\left(\begin{array}{c}
-B \\ C
\end{array}\right)-2\lambda\left(\begin{array}{c}
-B \\ C
\end{array}\right)+2y\lambda\left(\begin{array}{c}
u_y\\v_y
\end{array}\right)
+2\big((2x\lambda \lambda_y-\lambda_t)-A_0\big)
\left(\begin{array}{c}
-u\\v
\end{array}\right),
\end{align}
\end{subequations}
where $A_0$ is a constant and recursion operator
\begin{align}
\label{L}
L=\left(\begin{array}{cc}
-\partial & 0\\
0 & \partial
\end{array}\right)+2\left(\begin{array}{c}
u \\ -v
\end{array}\right)\partial^{-1}(v,u).
\end{align}
Taking $2\lambda \lambda_y-\lambda_t=\lambda^2$, $A_0=0$ and expanding $(B,C)^{\st}$ into polynomial as
\begin{align}
\label{BC}
\left(\begin{array}{c}
B\\C
\end{array}\right)=\sum_{j=1}^{2}\left(\begin{array}{c}
b_j\\c_j
\end{array}\right)\lambda^{2-j},
\end{align}
by imposing some special choices on $(b_j,c_j)^{\st}$, from \eqref{q-r} one can derive the
NBS-AKNS equation \eqref{n-(2+1)-DBS}. Thus \eqref{n-(2+1)-DBS-Lp}
supplies Lax representation for system \eqref{n-(2+1)-DBS}, where
\begin{subequations}
\label{ABC}
\begin{align}
& A=-x\lambda^2+\frac{1}{2}xuv+\frac{1}{2}\partial^{-1}_x(uv)+y\partial^{-1}_x(uv)_y, \\
& B=xu\lambda-\frac{1}{2}xu_x-\frac{1}{2}u-yu_y, \\
& C=xv\lambda+\frac{1}{2}xv_x+\frac{1}{2}v+yv_y.
\end{align}
\end{subequations}

\subsection{Double Wronskian solutions}

Introducing the dependent variable transformations
\begin{align}
\label{uv-tran}
u=\frac{g}{f}, \qquad v=-\frac{h}{f},
\end{align}
the NBS-AKNS equation \eqref{n-(2+1)-DBS} can be transformed into the following bilinear form
\begin{subequations}
\label{bi-sys1}
\begin{align}
& \label{bi-sys1-a} (D_t+\frac{x}{2}D^2_x+yD_x D_y)g\cdot f=-g_xf, \\
& \label{bi-sys1-b} (D_t-\frac{x}{2}D^2_x-yD_x D_y)h\cdot f=h_xf, \\
& \label{bi-sys1-c} D^2_xf \cdot f-2gh=0,
\end{align}
\end{subequations}
where $D$ is the well-known Hirota bilinear operator \cite{Hirota-book} defined by
\begin{align*}
D_t^m D_x^n D_y^l f\cdot g=(\partial_t-\partial_{t'})^m(\partial_x-\partial_{x'})^n(\partial_y-\partial_{y'})^l f(t,x,y)g(t',x',y')|_{t'=t,x'=x,y'=y}.
\end{align*}

Double Wronski determinant solutions of the bilinear form \eqref{bi-sys1} can be summarized in the
following theorem.
\begin{Thm}
\label{Thm-AKNS-solu}
The double Wronski determinants
\begin{subequations}
\label{fgh1}
\begin{align}
& f=|\phi,\partial_x\phi,\ldots,\partial^{n}_x\phi;\psi,\partial_x\psi,\ldots,\partial^{m}_x\psi|=|\wh{\phi}^{(n)};\wh{\psi}^{(m)}|, \\
& g=2|\phi,\partial_x\phi,\ldots,\partial^{n+1}_x\phi;\psi,\partial_x\psi,\ldots,\partial^{m-1}_x\psi|=2|\wh{\phi}^{(n+1)};\wh{\psi}^{(m-1)}|, \\
& h=2|\phi,\partial_x\phi,\ldots,\partial^{n-1}_x\phi;\psi,\partial_x\psi,\ldots,\partial^{m+1}_x\psi|=2|\wh{\phi}^{(n-1)};\wh{\psi}^{(m+1)}|,
\end{align}
\end{subequations}
constituted by $\phi=(\phi_1, \phi_2, \ldots, \phi_{n+m+2})^{\st}$ and $\psi=(\psi_1, \psi_2, \ldots, \psi_{n+m+2})^{\st}$
solve the bilinear system \eqref{bi-sys1}, provided that $\phi$ and $\psi$ satisfy the following condition equation set
\begin{subequations}
\label{CES1}
\begin{align}
& \label{CES1-a} \phi_{x}=K(t)\phi, \quad \phi_{y}=\phi_{xx},\quad  \phi_{t}=-2y\phi_{xxx}-x\phi_{xx}+n\phi_{x}, \\
& \label{CES1-b} \psi_{x}=-K(t)\psi, \quad \psi_{y}=-\psi_{xx},\quad  \psi_{t}=-2y\psi_{xxx}+x\psi_{xx}-m\psi_{x},
\end{align}
\end{subequations}
respectively, where $K(t)$ is an $(n+m+2)\times (n+m+2)$ matrix satisfying $K_{t}(t)=-K^2(t)$.
\end{Thm}

Here we skip the proof since it is similar to the one given in \cite{Yao}.
Considering the condition equation set \eqref{CES1}, we know
\begin{align}
\label{phsi-solu1}
\phi=\left(-K(t)\right)^{-n}e^{K(t)x+K^2(t)y} \alpha, \quad
\psi=\left(-K(t)\right)^{-m}e^{-K(t)x-K^2(t)y} \beta,
\end{align}
where $\alpha=(\alpha_1,\alpha_2,\ldots,\alpha_{n+m+2})^{\st}$ and $\beta=(\beta_1,\beta_2,\ldots,\beta_{n+m+2})^{\st}$
are two constant column vectors. The complex local reduction of the NBS-AKNS equation \eqref{n-(2+1)-DBS}
has been discussed in \cite{Yao}.
In the following two sections, we, respectively, investigate real and complex nonlocal reductions
of the system \eqref{n-(2+1)-DBS}.

	
\section{Real nonlocal reduction of the system \eqref{n-(2+1)-DBS}} \label{sec-3}

In this section, we shall consider real nonlocal reduction of the system \eqref{n-(2+1)-DBS}.
We show that the real nonlocal reduced
equations include reverse-$(y,t)$ type and reverse-$(x,y)$ type.
By imposing suitable constraints on the pair $\phi$ and $\psi$ in the double Wronski determinant,
we derive the formal expression of solution for the reduced equations. And moreover,
we present soliton solutions and Jordan-block solutions in terms of the eigenvalue structure of matrix $K$.
Dynamics of some obtained solutions are also discussed.

\subsection{Reduction procedure}\label{sec-3.1.1}

For the NBS-AKNS equation \eqref{n-(2+1)-DBS}, it allows a real nonlocal reduction 	
\begin{align}
\label{real-re1}
v(x,y,t)=\delta u(\sigma x,-y,-\sigma t),\quad \sigma,~~\delta=\pm 1.
\end{align}
In this case the equation \eqref{n-(2+1)-DBS} yields
\begin{align}
\label{nn-BSE-R}
u_t= &-y(u_{xy}-2\delta u\partial^{-1}_{x}(uu(\sigma x,-y,-\sigma t))_y)-\frac{x}{2}(u_{xx}-2\delta u^2u(\sigma x,-y,-\sigma t)) \nn \\
& -u_{x}+\delta u\partial^{-1}_{x}(uu(\sigma x,-y,-\sigma t)),
\end{align}
which is the RNNBS-NLS equation. Equation \eqref{nn-BSE-R} is preserved under transformation $u\rightarrow -u$.
Besides, equation \eqref{nn-BSE-R} with $(\sigma,\delta)=(\pm 1, 1)$ and
with $(\sigma,\delta)=(\pm 1,-1)$ can be transformed into each other by taking $u\rightarrow \pm iu$.
When $\sigma=1$ \eqref{nn-BSE-R} is reverse-$(y,t)$ type and when $\sigma=-1$ \eqref{nn-BSE-R} is reverse-$(x,y)$ type.
Compared with Ref. \cite{WWZ-CTP}, we know that both of the
isospectral and nonisospectral (2+1)-dimensional breaking soliton AKNS equations have
reverse-$(y,t)$ type nonlocal reduction. There is no reverse-$(x,y,t)$ nonlocal reduction
for the NBS-AKNS equation \eqref{n-(2+1)-DBS}. The reverse-$(x,y)$ nonlocal reduction
which is valid for the nonlocal reduction of the NBS-AKNS equation \eqref{n-(2+1)-DBS},
is also not admitted in the nonlocal reduction of the
isospectral (2+1)-dimensional breaking soliton AKNS equation.

Our subsequent work is to implement the reduction procedures on the double Wronskian solutions, by which solutions for the
RNNBS-NLS equation \eqref{nn-BSE-R} can be obtained from those of unreduced
NBS-AKNS equation mentioned in Theorem \ref{Thm-AKNS-solu}.
For this purpose, we take $m=n$. Double Wronskian solutions to equation
\eqref{nn-BSE-R} can be summarized in the following theorem.	
\begin{Thm}
\label{Thm-BSER-so}
Double Wronskian solutions of the RNNBS-NLS equation \eqref{nn-BSE-R} are given by
\begin{align}
\label{u-v1}
u=\frac{g}{f},\quad f=|\wh{\phi}^{(n)}; \wh{\psi}^{(n)}|,\quad g=2|\wh{\phi}^{(n+1)}; \wh{\psi}^{(n-1)}|,
\end{align}
in which $\phi$ and $\psi$ are the $2(n+1)$-th order column vectors defined by \eqref{phsi-solu1}, and satisfy the following relation
\begin{align}
\label{real-con}
\psi(x,y,t)=(-\sigma)^{n}T\phi(\sigma x,-y,-\sigma t),
\end{align}
where $T\in \mathbb{C}^{2(n+1)\times 2(n+1)}$ is a constant matrix satisfying the determining equations
\begin{align}
\label{real-kT}
K(t)T+\sigma T K(-\sigma t)=0,\quad T^2=\sigma\delta I,
\end{align}
and we require $\beta=T\alpha$.
\end{Thm}
	
\begin{proof}
Under the first determining equation in \eqref{real-kT}, we know
\begin{align}
\psi(x,y,t) &=\left(-K(t)\right)^{-n}e^{-K(t)x-K^2(t)y} \beta \nn\\
&=\left(\sigma TK(-\sigma t) T^{-1})\right)^{-n}e^{\sigma TK(-\sigma t) T^{-1}x-TK^2(-\sigma t)T^{-1}y} \beta \nn \\
&=(-\sigma)^{-n} T \left(-K(-\sigma t)\right)^{-n} e^{\sigma K(-\sigma t) x+K^2(-\sigma t)(-y)} T^{-1}\beta \nn\\
&=(-\sigma)^{n}T\phi(\sigma x,-y,-\sigma t),
\end{align}	
which means that relation \eqref{real-con} holds. To identify the connections among variables $f$, $g$ and $h$, we have to rewrite them as
\begin{subequations}
\label{fgh-x}
\begin{align}
& \label{f-x}
f=|\wh{\phi}^{(n)}; \wh{\psi}^{(n)}|=|\wh{\phi}^{(n)}(x,y,t)_{[x]};(-\sigma)^{n}T\wh{\phi}^{(n)}(\sigma x,-y,-\sigma t)_{[x]}|, \\
& \label{g-x}
g=2|\wh{\phi}^{(n+1)}; \wh{\psi}^{(n-1)}|=2|\wh{\phi}^{(n+1)}(x,y,t)_{[x]};(-\sigma)^{n}T\wh{\phi}^{(n-1)}(\sigma x,-y,-\sigma t)_{[x]}|, \\
& \label{h-x}
h=2|\wh{\phi}^{(n-1)}; \wh{\psi}^{(n+1)}|=2|\wh{\phi}^{(n-1)}(x,y,t)_{[x]};(-\sigma)^{n}T\wh{\phi}^{(n+1)}(\sigma x,-y,-\sigma t)_{[x]}|,
\end{align}
\end{subequations}
where we have introduced notation
\begin{align}
\wh{\phi}^{(s)}(ax)_{[bx]}=(\phi(ax), \partial_{bx}\phi(ax),\ldots,\partial^s_{bx}\phi(ax))_{2(n+1)\times 2(s+1)}, \quad a,~b=\pm 1.
\end{align}
Noticing the second determining equation in \eqref{real-kT}, we find
\begin{align}
\label{relation-ff}
f(\sigma x,-y,-\sigma t)=&|\wh{\phi}^{(n)}(\sigma x,-y,-\sigma t)_{[\sigma x]}; (-\sigma)^{n}T\wh{\phi}^{(n)}(x,y,t)_{[\sigma x]}| \nn\\
=& (-1)^{(n+1)^2}|T||\wh{\phi}^{(n)}(x,y,t)_{[x]};(-\sigma)^{n}T^{-1}\wh{\phi}^{(n)}(\sigma x,-y,-\sigma t)_{[x]}|\nn \\	
=&(-1)^{(n+1)^2}|T|(\sigma \delta)^{(n+1)}|\wh{\phi}^{(n)}(x,y,t)_{[x]};(-\sigma)^{n}T^{-1}\wh{\phi}^{(n)}(\sigma x,-y,-\sigma t)_{[x]}| \nn \\
=&(-1)^{(n+1)^2}|T|(\sigma\delta)^{n+1}f(x,y,t),
\end{align}
as well as
\begin{align}
\label{relation-gh}
g(\sigma x,-y,-\sigma t)=&2|\wh{\phi}^{(n+1)}(\sigma x,-y,-\sigma t)_{[\sigma x]}; (-\sigma)^{n}T\wh{\phi}^{(n-1)}(x,y,t)_{[\sigma x]}| \nn\\
=& 2(-1)^{n(n+2)}|T|\sigma^{2n+1}|\wh{\phi}^{(n-1)}(x,y,t)_{[x]};(-\sigma)^{n}T^{-1}\wh{\phi}^{(n+1)}(\sigma x,-y,-\sigma t)_{[x]}|\nn \\	
=& 2(-1)^{n(n+2)}|T|\sigma^{2n+1}(\sigma\delta)^{n+2}|\wh{\phi}^{(n-1)}(x,y,t)_{[x]};(-\sigma)^{n}T^{-1}\wh{\phi}^{(n+1)}(\sigma x,-y,-\sigma t)_{[x]}| \nn \\
=&(-1)^{n(n+2)}|T|\sigma(\sigma \delta)^{n+2}h(x,y,t).
\end{align}
Thus utilizing transformation \eqref{uv-tran}, we have
\begin{align*}
u(\sigma x,-y,-\sigma t)=&\frac{g(\sigma x,-y,-\sigma t)}{f(\sigma x,-y,-\sigma t)}
=\frac{(-1)^{n(n+2)}|T|\sigma(\sigma \delta)^{n+2}h(x,y,t)}{(-1)^{(n+1)^2}|T|(\sigma\delta)^{n+1}f(x,y,t)}\nn \\
=&-\delta \frac{h(x,y,t)}{f(x,y,t)}=\delta v(x,y,t),
\end{align*}
which coincides with the reduction \eqref{real-re1} for the NBS-AKNS equation \eqref{n-(2+1)-DBS}.
Therefore, we complete the verification.
\end{proof}

\noindent {\bf Remark 1}: {\it From the Theorem \ref{Thm-BSER-so}, we know that
the formal double Wronskian solution to the RNNBS-NLS equation \eqref{nn-BSE-R} is
expressed by $u=\frac{g}{f}$ with
\begin{align*}
f=& |\wh{\phi}^{(n)}(x,y,t);(-\sigma)^{n}T\wh{\phi}^{(n)}(\sigma x,-y,-\sigma t)|,\nn \\
g=& 2|\wh{\phi}^{(n+1)}(x,y,t);(-\sigma)^{n}T\wh{\phi}^{(n-1)}(\sigma x,-y,-\sigma t)|,
\end{align*}
where $\phi$ is given by \eqref{phsi-solu1} and $K(t)$ and $T$ satisfy the constraint relations \eqref{real-kT}}.
	
\subsection{Some examples of solutions} \label{sec-3.2.2}

To give the explicit expressions of solution $u$, one needs to solve the determining
equations \eqref{real-kT}, where $K(t)$ satisfies $K_{t}(t)=-K^2(t)$. To do so, we divide
$K(t)$ and $T$ as
\begin{align}
\label{real-kT-def}
K(t)=\left(\begin{array}{cc}
K_1(t) & \bm 0 \\
\bm 0 & K_2(t) \\
\end{array}\right),
\quad
T=\left(\begin{array}{cc}
T_1 & T_2 \\
T_3 & T_4 \\
\end{array}\right),
\end{align}
with $K_i(t),~~T_j \in \mathbb{C}^{(n+1)\times(n+1)},~i=1,2,j=1,2,3,4$.
Substituting \eqref{real-kT-def} into equations \eqref{real-kT} and we can directly get the solutions for $K(t)$ and $T$.
We list solutions to \eqref{real-kT} with different $(\sigma,\delta)$ in Table 1.

\begin{center}
\footnotesize \setlength{\tabcolsep}{8pt}
\renewcommand{\arraystretch}{1.5}
\begin{tabular}[htbp]{|l|l|l|l|}
\hline
$(\sigma,\delta)$& $K(t)$ & $T$\\
\hline
$(1,-1)$ &\eqref{real-kT-def} \mbox{with} $K_2(t)=-K_1(-t)$ & $T_1=T_4=\bm 0, \quad T_3=-T_2=I$ \\
\hline
$(1,1)$  & \eqref{real-kT-def} \mbox{with} $K_2(t)=-K_1(-t)$ & $T_1=T_4=\bm 0, \quad T_3=T_2=I$ \\
\hline
$(-1,-1)$ &\eqref{real-kT-def} & $T_1=-T_4=I, \quad T_3=T_2=\bm 0$ \\
\hline
$(-1,1)$ & \eqref{real-kT-def} & $T_1=-T_4=iI, \quad T_3=T_2=\bm 0$ \\ \hline
\end{tabular}
\end{center}
\begin{center}
\begin{minipage}{7cm}{\footnotesize
{Table 1. \emph{$K(t)$ and $T$ for equation \eqref{real-kT}.}}}
\end{minipage}
\end{center}
	
In what follows, we consider soliton solutions and Jordan-block solutions for the RNNBS-NLS equation \eqref{nn-BSE-R}.
For the sake of brevity, we introduce some notations
\begin{align}
\label{pqtheta}
p_i=(t-a_i)^{-1}, \quad q_i=(t+a_i)^{-1},\quad r_i=(t-b_i)^{-1},\quad \theta_i=\frac{\wt{\alpha}_{i}}{\alpha_i},
\end{align}
where ${\wt{\alpha}_{i}=\alpha_{n+1+i}}$ and ${a_i,~b_i}$ with $i=1,2,\ldots,n+1$ are real constants.

\noindent \textbf{Soliton solutions:} In this situation, we set $K_1(t)$ as a diagonal matrix
\begin{align}
K_1=\mbox{Diag}(p_1,~p_2,~\ldots,~p_{n+1}).
\end{align}
When $(\sigma,\delta)=(1, \pm 1)$, matrix $K_2(t)$ has to be the form of
$K_2=\mbox{Diag}(q_1,~q_2,~\ldots,~q_{n+1})$. While in the case of $(\sigma,\delta)=(-1, \pm 1)$, matrix $K_2(t)$ can be independent of
matrix $K_1(t)$, and we take $K_2=\mbox{Diag}(r_1,~r_2,~\ldots,~r_{n+1})$.
	
\noindent {\bf Remark 2}: {\it Because of the block structure of matrix $T$, one can easily observe that $\alpha$ can
be gauged to be $(1,1,\ldots,1;1,1,\dots,1)^T$. It implies that solutions obtained for the case
$(\sigma=-1,\delta=\pm 1)$ are independent of phase parameters in $\alpha$, i.e., the
initial phase has always to be zero. Soliton solutions and Jordan-block solutions listed
below can demonstrate this character.}
	
When $n=0$, we list the one-soliton solution
\begin{subequations}
\label{u-one-solu3}
\begin{align}
& \label{u(1,-1)}
u_{\sigma=1,\delta=-1}=(q_1-p_1)e^{(p_1^2+q_1^2)y}\sech((q_1-p_1)x+\ln\theta_1),\\
& \label{u(1,1)}
u_{\sigma=1,\delta=1}=(p_1-q_1)e^{(p_1^2+q_1^2)y}\csch((q_1-p_1)x+\ln\theta_1),\\
& \label{u(-1,-1)}
u_{\sigma=-1,\delta=-1}=\frac{2(p_1-r_1)}{e^{-2p_1x-2p_1^2y}+e^{-2r_1x-2r_1^2y}},\\
& \label{u(-1,1)}
u_{\sigma=-1,\delta=1}=-\frac{2i(p_1-r_1)}{e^{-2p_1x-2p_1^2y}+e^{-2r_1x-2r_1^2y}}.
\end{align}
\end{subequations}
Let's briefly identify the dynamics of solution \eqref{u(1,-1)} .
This is a moving wave with an initial phase $\ln\theta_1$ and a $(y,t)$-dependent
amplitude $(q_1-p_1)e^{(p_1^2+q_1^2)y}$. Besides, $x(t)=(p_1-q_1)^{-1}\ln\theta_1$ and $\frac{dx(t)}{dt}=(q_1^2-p_1^2)\ln\theta_1$,
respectively, denote time-varying top trajectory and propagation velocity. As $\theta_1=1$, the wave \eqref{u(1,-1)} is stationary.
Fixing $y$, the top trace of one-soliton in coordinate frame $\{x,t\}$ is depicted as Fig. 1.

\begin{center}
\begin{picture}(120,100)
\put(-120,-23){\resizebox{!}{4cm}{\includegraphics{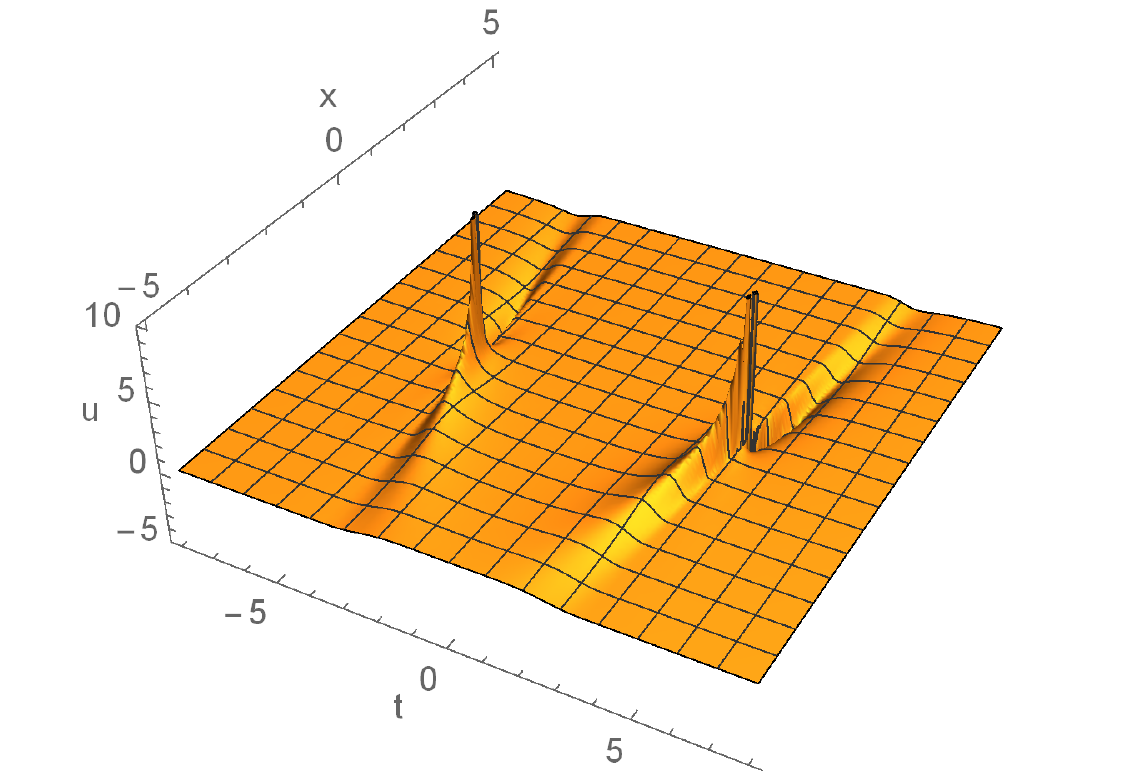}}}
\put(100,-23){\resizebox{!}{4cm}{\includegraphics{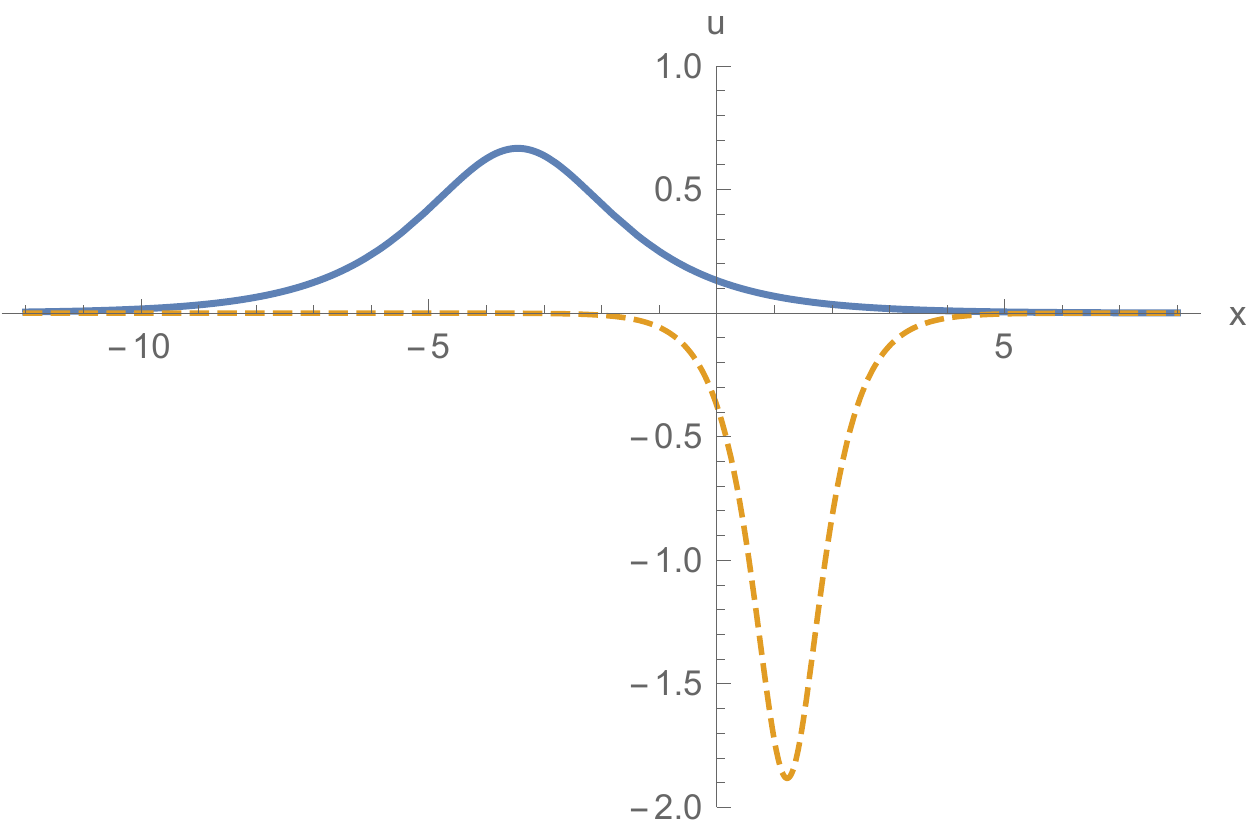}}}
\end{picture}
\end{center}
\vskip 20pt
\begin{center}
\begin{minipage}{15cm}{\footnotesize
\quad\qquad\qquad\qquad\qquad\qquad(a)\qquad\qquad\qquad\qquad\qquad\qquad\qquad\qquad\qquad\qquad\quad \qquad \quad (b)}
\end{minipage}
\end{center}
\vskip 10pt
\begin{center}
\begin{picture}(120,80)
\put(-120,-23){\resizebox{!}{4cm}{\includegraphics{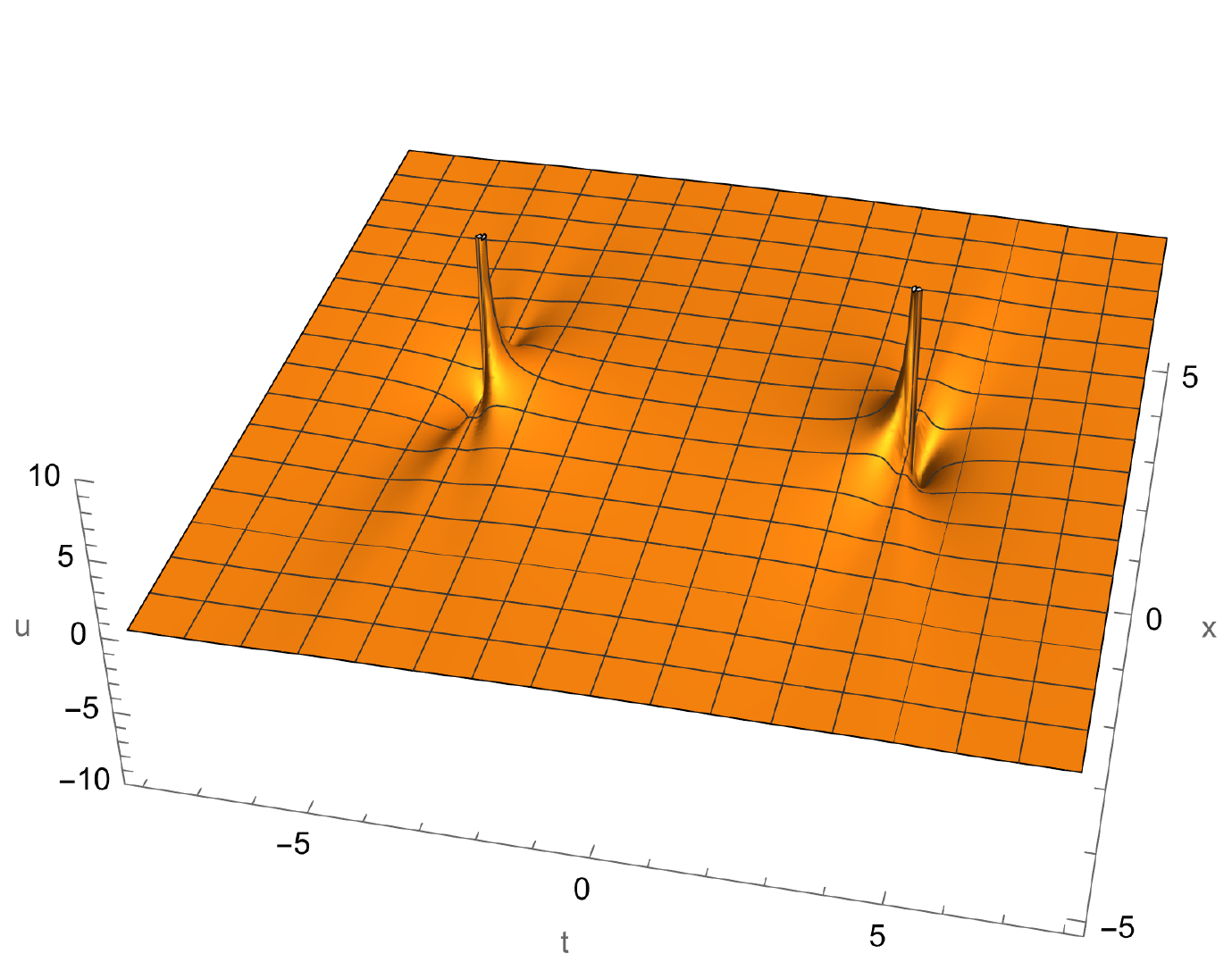}}}
\put(100,-23){\resizebox{!}{4cm}{\includegraphics{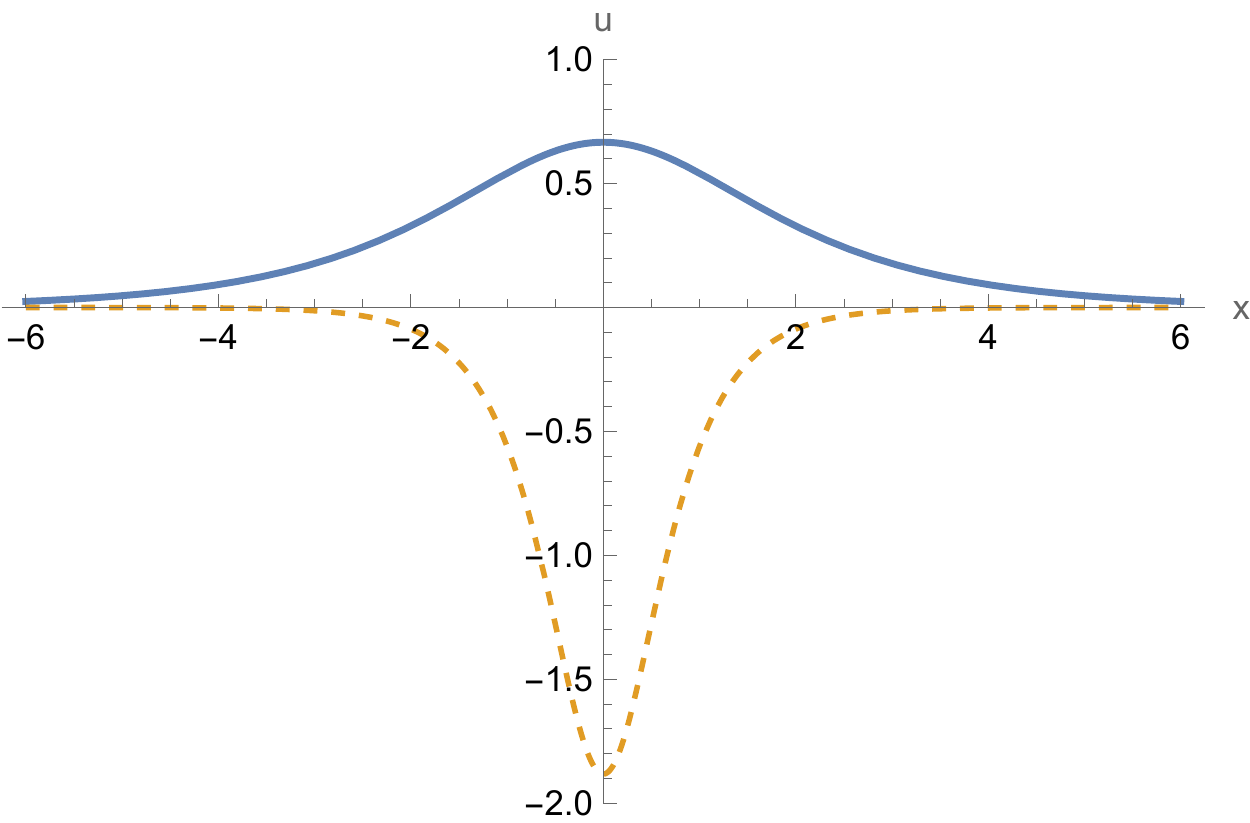}}}
\end{picture}
\end{center}
\vskip 20pt
\begin{center}
\begin{minipage}{15cm}{\footnotesize
\quad\qquad\qquad\qquad\qquad\qquad\qquad(c)\qquad\qquad\qquad\qquad\qquad\qquad\qquad\qquad\qquad\qquad\quad \qquad \quad (d)\\
{\bf Fig. 1} shape and motion of $u$ given by \eqref{u(1,-1)} for $a_1=4$ and $y=0$.
(a) a moving wave for $\theta_1=10$.
(b) waves in solid line and dotted line stand for plot (a) at $t=2$ and $t=4.5$, respectively.
(c) a stationary wave for $\theta_1=1$.
(d) waves in solid line and dotted line stand for plot (c) at $t=2$ and $t=4.5$, respectively.}
\end{minipage}
\end{center}

When $n=1$, one can get the two-soliton solutions, which read
\begin{subequations}
\label{u-two-solu}
\begin{align}
& \label{2-u(1,-1)}
u_{\sigma=1,\delta=-1}=\frac{\alpha_1\wt{\alpha}_{1}A_1(\alpha_2^2A_2+\wt{\alpha}_{2}A_2^{p\leftrightarrow q})
+\alpha_2\wt{\alpha}_{2}A_{1,{1\leftrightarrow 2}}(\alpha_{1}^2A_{2,{1\leftrightarrow 2}}+{\wt{\alpha}}_{1}^2\bar{A_2})}
{\alpha_{1}^2\alpha_{2}^2A_3+\wt{\alpha}_{1}^2\wt{\alpha}_{2}^2A^{p\leftrightarrow q}_{3}+\wt{\alpha}_{1}^2\alpha_{2}^2A_4
+\alpha_{1}^2{\wt{\alpha}}_{2}^2A_{4,{1\leftrightarrow 2}}-A_5},\\
& \label{2-u(1,1)}
u_{\sigma=1,\delta=1}=\frac{\alpha_1\wt{\alpha}_{1}A_1(\alpha_2^2A_2-\wt{\alpha}_{2}A_2^{p\leftrightarrow q})
+\alpha_2\wt{\alpha}_{2}A_{1,{1\leftrightarrow 2}}(\alpha_{1}^2A_{2,{1\leftrightarrow 2}}-{\wt{\alpha}}_{1}^2\bar{A_2})}
{\alpha_{1}^2\alpha_{2}^2A_3+\wt{\alpha}_{1}^2\wt{\alpha}_{2}^2A^{p\leftrightarrow q}_{3}-
\wt{\alpha}_{1}^2\alpha_{2}^2A_4-\alpha_{1}^2\wt{\alpha}_{2}^2A_{4,{1\leftrightarrow 2}}+A_5},\\
& \label{2-u(-1,-1)}
u_{\sigma=-1,\delta=-1}=-\frac{2B_1(B_2-B_3)-2B_1^{p\leftrightarrow r}(B_2^{p\leftrightarrow r}-B_3^{p\leftrightarrow r})}{B_4+B_5-B_6}, \\
& \label{2-(-1,1)}
u_{\sigma=-1,\delta=1}=\frac{2iB_1(B_2-B_3)-2iB_1^{p\leftrightarrow r}(B_2^{p\leftrightarrow r}-B_3^{p\leftrightarrow r})}{B_4+B_5-B_6},
\end{align}	
\end{subequations}
in which
\begin{subequations}
\begin{align}
&A_1=2p_1q_1(p_1-q_1)e^{(p_1+q_1)x+(p_1^2+q_1^2)y}, \quad A_2=q_2^2(p_1-p_2)(p_2-q_1)e^{2p_2x},\\
&A_3=q_1^2q_2^2(p_1-p_2)^2e^{2(p_1+p_2)x}, \quad A_4=p_1^2q_2^2(p_2-q_1)^2e^{2(p_2+q_1)x},\\
&A_5=2p_1p_2q_1q_2\alpha_{1}\alpha_{2}\wt{\alpha}_{1}\wt{\alpha}_{2}(p_1-q_1)(p_2-q_2)e^{(p_1+p_2+q_1+q_2)x}\cosh((p_1^2-p_2^2+q_1^2-q_2^2)y),\\
&B_1=(p_1-p_2)e^{2(p_1+p_2)x+2(p_1^2+p_2^2)y}, \quad B_2=(p_1-r_1)(p_2-r_1)e^{2r_1x+2r_1^2y},\\
&B_3=(p_1-r_2)(p_2-r_2)e^{2r_2x+2r_2^2y},\\
&B_4=(p_1-r_1)(p_2-r_2)(e^{2(p_1+r_1)x+2(p_1^2+r_1^2)y}+e^{2(p_2+r_2)x+2(p_2^2+r_2^2)y}),\\
&B_5=(p_1-p_2)(r_1-r_2)(e^{2(p_1+p_2)x+2(p_1^2+p_2^2)y}+e^{2(r_1+r_2)x+2(r_1^2+r_2^2)y}),\\
&B_6=(p_2-r_1)(p_1-r_2)(e^{2(p_2+r_1)x+2(p_2^2+r_1^2)y}+e^{2(p_1+r_2)x+2(p_1^2+r_2^2)y}),
\end{align}	
\end{subequations}
where $G_j^{p\leftrightarrow q}$ means that $p_{i}$ and $q_{i}$ in $G_j$ exchange each other;
$G_{j,{1\leftrightarrow 2}}$ represents the exchange of subscript, e.g. $p_1 \rightleftharpoons p_2$;
$\bar{G_j}$ contains both transformations of $G_j^{p\leftrightarrow q}$ and $G_{j,{1\leftrightarrow 2}}$.

We next pay attention to the two-soliton solutions \eqref{2-u(1,-1)}. Since $p_{i}$ and $q_{i}$ are functions of $t$, it is
intractable to make asymptotic analysis as usual \cite{Hie-book}. Here, we only depict \eqref{2-u(1,-1)} in Fig. 2.

\begin{center}
\begin{picture}(120,100)
\put(-160,-23){\resizebox{!}{4cm}{\includegraphics{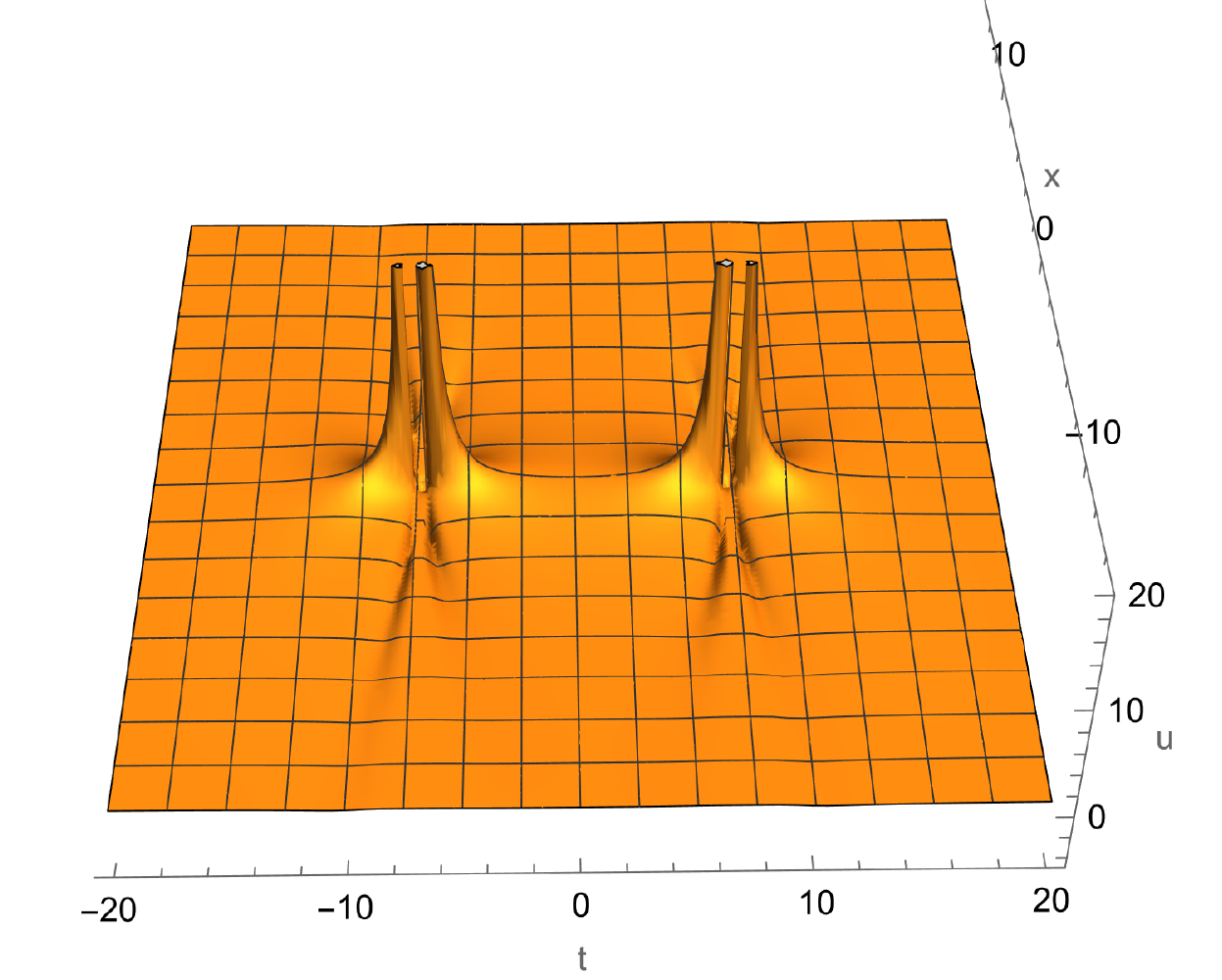}}}
\put(-10,-23){\resizebox{!}{4cm}{\includegraphics{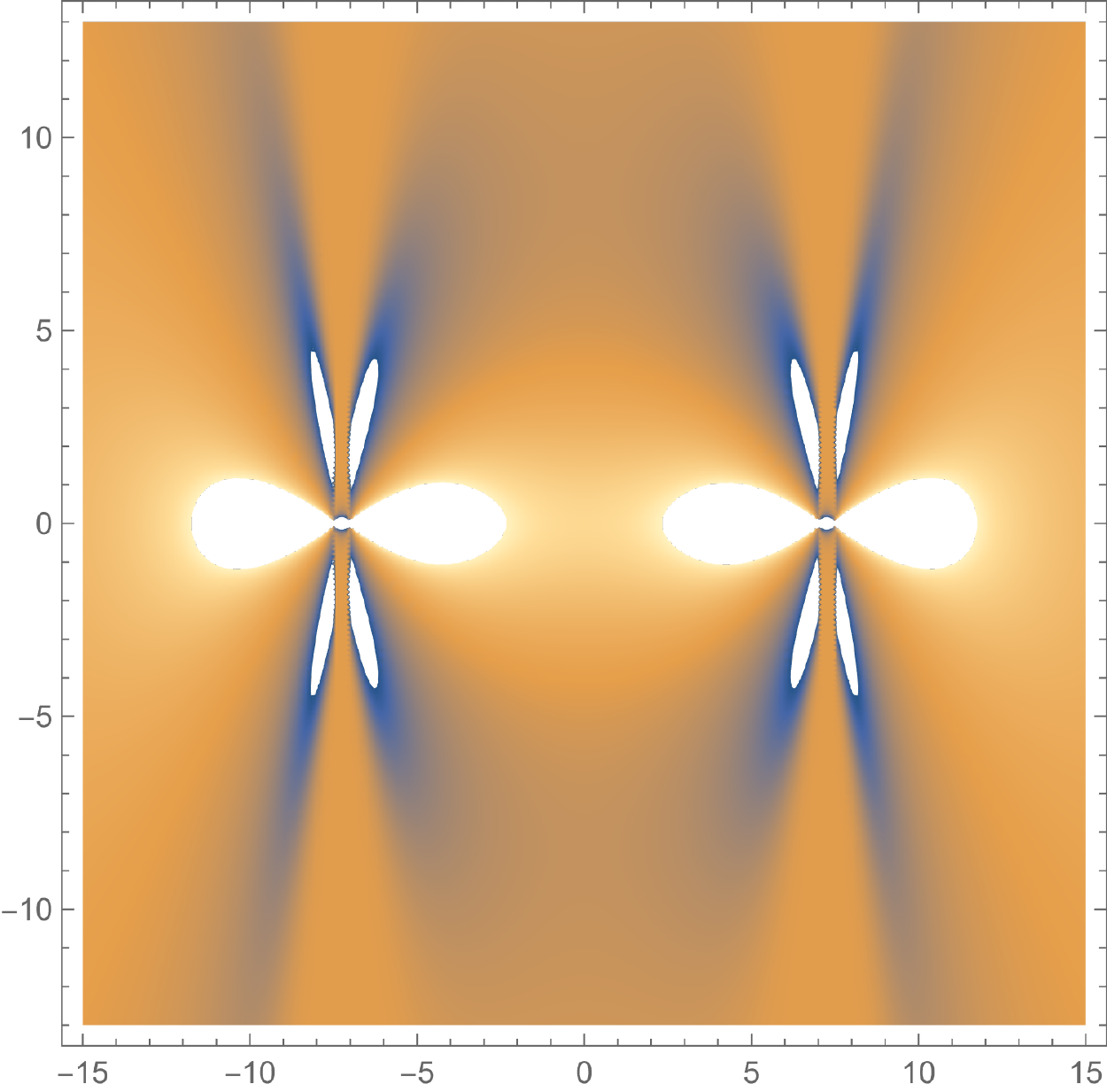}}}
\put(130,-23){\resizebox{!}{4cm}{\includegraphics{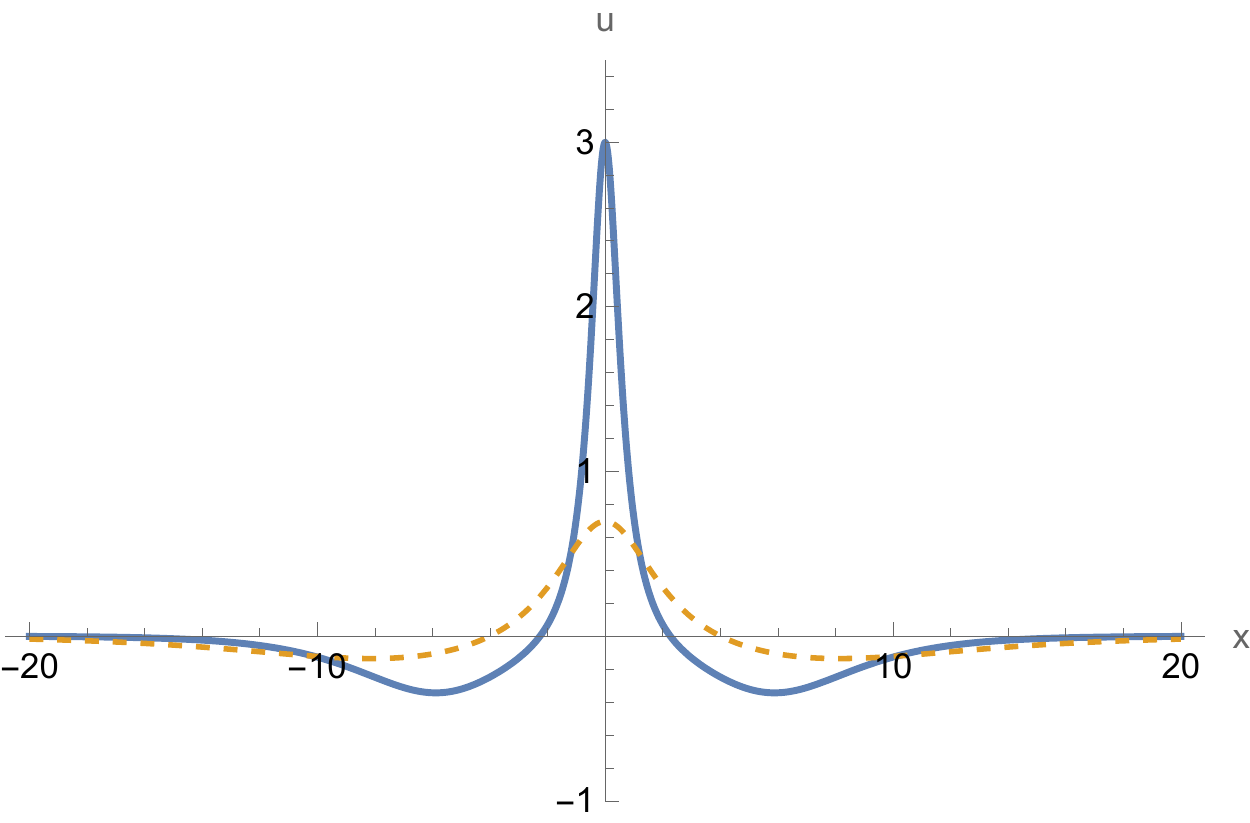}}}
\end{picture}
\end{center}
\vskip 20pt
\begin{center}
\begin{minipage}{16cm}{\footnotesize
\quad\qquad\qquad\qquad(a)\quad\qquad\qquad\qquad \qquad\quad \qquad \qquad \quad (b) \qquad\qquad \qquad \qquad\qquad\qquad\qquad \quad (c)}
\end{minipage}
\end{center}
\vskip 10pt
\begin{center}
\begin{picture}(120,80)
\put(-160,-23){\resizebox{!}{4cm}{\includegraphics{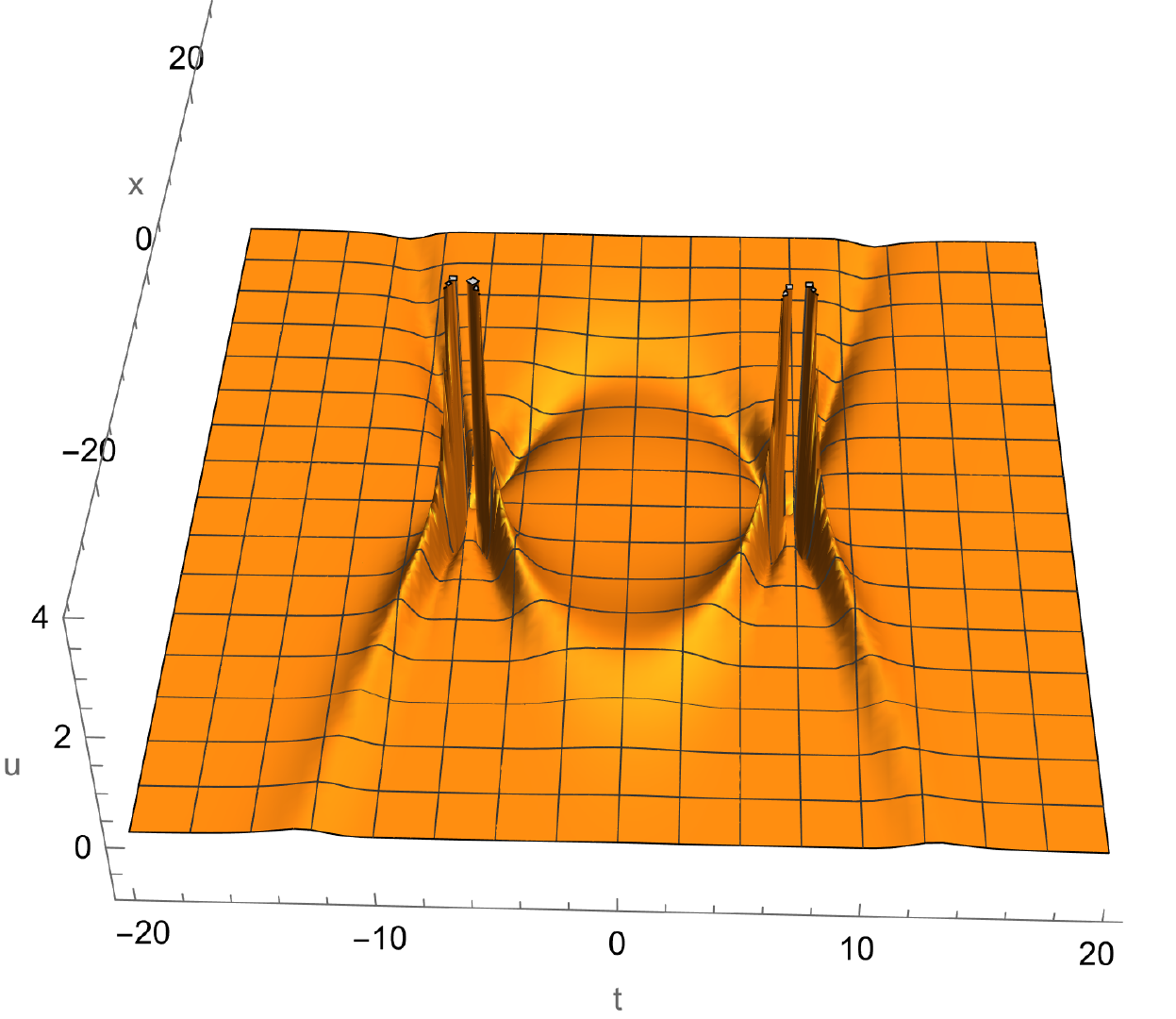}}}
\put(-10,-23){\resizebox{!}{4cm}{\includegraphics{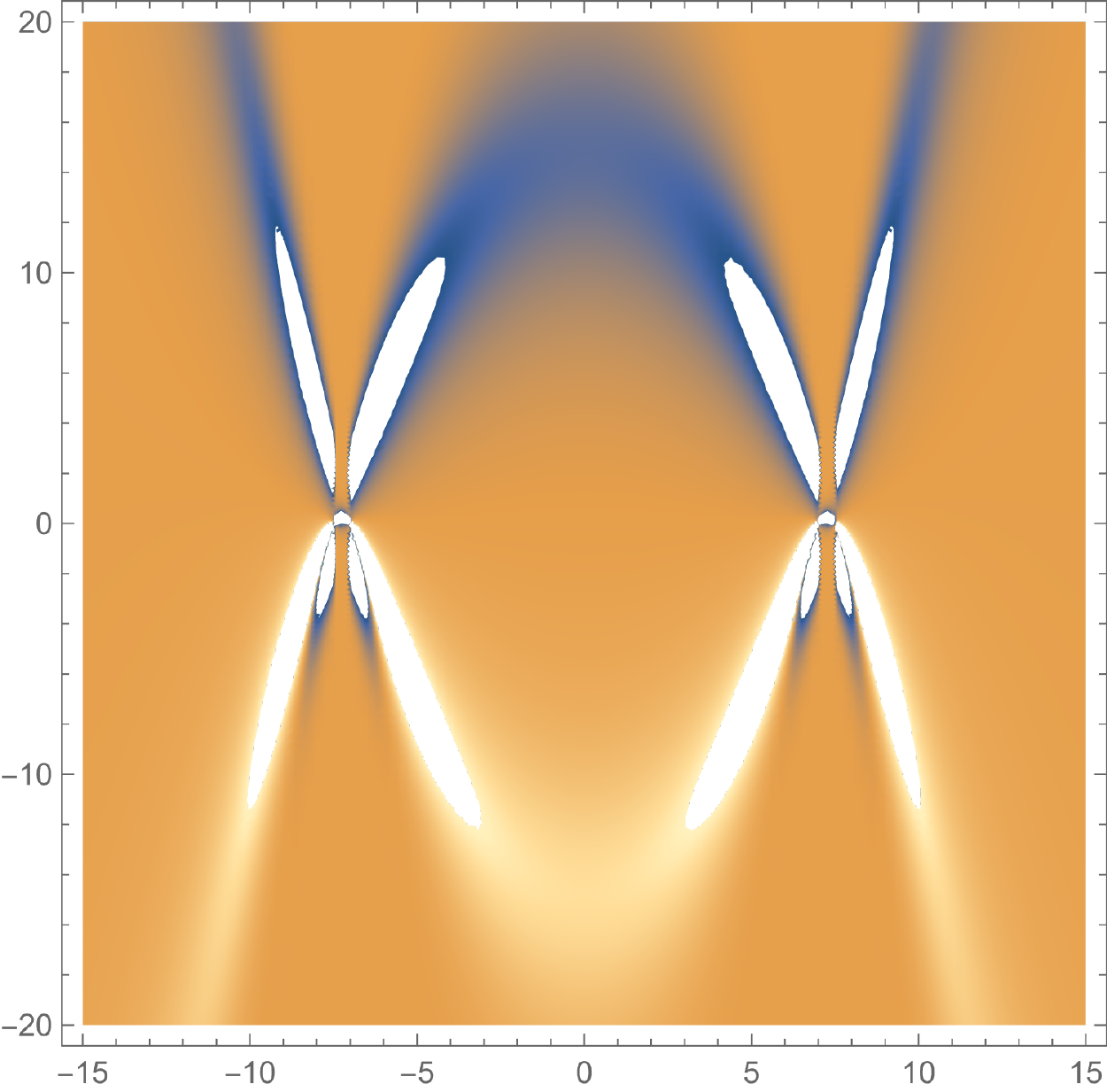}}}
\put(130,-23){\resizebox{!}{4cm}{\includegraphics{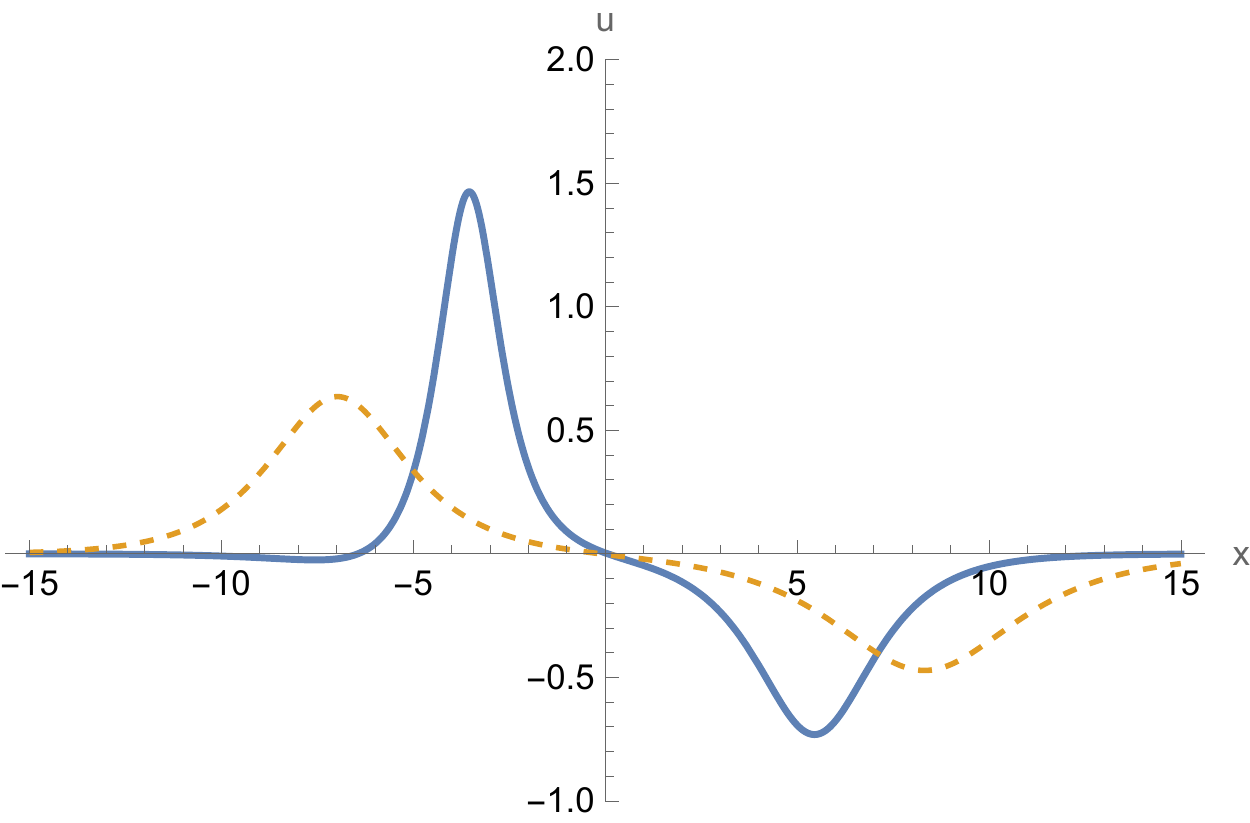}}}
\end{picture}
\end{center}
\vskip 20pt
\begin{center}
\begin{minipage}{16cm}{\footnotesize
\quad\qquad\qquad\qquad(d)\quad\qquad\qquad\qquad \qquad\quad \qquad \qquad \quad (e) \qquad\qquad \qquad \qquad\qquad\qquad\qquad \quad (f)\\
{\bf Fig. 2} (a) shape and motion of two-soliton solutions $u$ given by \eqref{2-u(1,-1)}
for $a_1=7,~a_2=7.5$, ~$\alpha_1=\alpha_2=\wt{\alpha}_{1}=\wt{\alpha}_{2}=1$ and $y=0$.
(b) a contour plot of (a) with range $x\in [-13, 13]$ and $t\in [-15,15]$.
(c) waves in solid line and dotted line stand for plot (a) at $t=5$ and $t=2$, respectively.
(d) shape and motion of two-soliton solutions $u$ given by \eqref{2-u(1,-1)} for $a_1=7,~a_2=7.5$,
$\alpha_1=1,~\alpha_2=2,~\wt{\alpha}_{1}=3, ~\wt{\alpha}_{2}=1$ and $y=0$.
(e) a contour plot of (d) with range $x\in [-20, 20]$ and $t\in [-15,15]$.
(f) waves in solid line and dotted line stand for plot (d) at $t=6$ and $t=5$, respectively.}
\end{minipage}
\end{center}

\noindent \textbf{Jordan-block solutions:}
Before giving the Jordan-block solutions, let's clarify lower triangular Toeplitz matrices, which are defined as
\begin{align*}
\mathcal{A}=\left(\begin{array}{cccccc}
\gamma_0 & 0    & 0   & \cdots & 0   & 0 \\
\gamma_1 & \gamma_0  & 0   & \cdots & 0   & 0 \\
\gamma_2 & \gamma_1  & \gamma_0 & \cdots & 0   & 0 \\
\vdots &\vdots &\cdots &\vdots &\vdots &\vdots \\
\gamma_{N-1} & \gamma_{N-2} & \gamma_{N-3}  & \cdots &  \gamma_1 & \gamma_0
\end{array}\right)_{N\times N}, \quad \gamma_j\in \mathbb{C}.
\end{align*}
Note that all the lower triangular Toeplitz matrices of the same order compose a commutative set in terms of matrix product.
Canonical form of such a matrix is a Jordan matrix.	For more properties of such matrices one can
refer to \cite{ZDJ-Wron,ZZSZ}.

We set $K_1(t)$ as a lower triangular Toeplitz matrix
\begin{align}
\label{Lam-Jor}
K_{1}(t)=(k_{s,j})_{(n+1)\times (n+1)}, \quad k_{s,j}=\Biggl\{
\begin{array}{ll}
\frac{1}{(s-j)!}\partial^{s-j}_{a_1}p_1,&~s\geq j,\\
0,&~s<j.
\end{array}
\end{align}
In the case of $(\sigma,\delta)=(1, \pm 1)$, we have
\begin{subequations}
\begin{align}
\phi_{j}=\Biggl\{
\begin{array}{ll}
\frac{\partial^{j-1}_{a_{1}}}{(j-1)!}\alpha_1(-p_1)^{-n}e^{p_1x+p_1^2y},& j=1,2,\dots,n+1,\\
\frac{\partial^{s-1}_{a_{1}}}{(s-1)!}\wt{\alpha}_1(-q_1)^{-n}e^{q_1x+q_1^2y},& j=n+1+s;~~s=1,2,\dots,n+1,\\
\end{array}
\end{align}
and in the case of $(\sigma,\delta)=(-1, \pm 1)$, we get
\begin{align}
\phi_{j}=\Biggl\{
\begin{array}{ll}
\frac{\partial^{j-1}_{a_{1}}}{(j-1)!}\alpha_1(-p_1)^{-n}e^{p_1x+p_1^2y},& j=1,2,\dots,n+1,\\
\frac{\partial^{s-1}_{b_{1}}}{(s-1)!}\wt{\alpha}_1(-r_1)^{-n}e^{r_1x+r_1^2y},& j=n+1+s;~~s=1,2,\dots,n+1.\\
\end{array}
\end{align}
\end{subequations}
Particularly, when $n=1$ the simplest Jordan-block solutions to the equation \eqref{nn-BSE-R} with different $(\sigma, \delta)$ read
\begin{subequations}
\label{lim-u}
\begin{align}
& \label{lim-u(1,-1)}
u_{\sigma=1,\delta=-1}=\frac{8a^2_1p_1^2q_1^2e^{E_2}(2x\sinh E_1+4y(p_1q_1)^{\frac{1}{2}}\sinh(E_1+\frac{1}{2}\ln{\frac{q_1}{p_1}})-\frac{1}{a_1p_1q_1}\cosh E_1)}
{\theta_1(\cosh(4a_1p_1q_1(x+2p_1y+q_1)+2\ln \theta_1))+E_3}, \\
& \label{lim-u(1,1)}
u_{\sigma=1,\delta=1}=\frac{8a^2_1p_1^2q_1^2e^{E_2}(2x\cosh E_1+4y(p_1q_1)^{\frac{1}{2}}\cosh(E_1+\frac{1}{2}\ln{\frac{q_1}{p_1}})-\frac{1}{a_1p_1q_1}\sinh E_1)}
{\theta_1(\cosh(4a_1p_1q_1(x+2p_1y+q_1)+2\ln \theta_1))-E_3}, \\
& \label{lim-u(-1,1)}
u_{\sigma=-1,\delta=1}=\frac{2(p_1-r_1)e^{-2x(a_1p_1^2+b_1r_1^2)}\left(e^{F_1}(1+F_3)+e^{F_2}(1+F_3^{p\leftrightarrow r})\right)}{1+\cosh(F_1-F_2)+2F_3F_3^{p\leftrightarrow r}},\\
& \label{lim-u(-1,-1)}
u_{\sigma=-1,\delta=-1}=\frac{-2i(p_1-r_1)e^{-2x(a_1p_1^2+b_1r_1^2)}\left(e^{F_1}(1+F_3)+e^{F_2}(1+F_3^{p\leftrightarrow r})\right)}{1+\cosh(F_1-F_2)+2F_3F_3^{p\leftrightarrow r}},
\end{align}
\end{subequations}
in which
\begin{subequations}
\begin{align}
&E_1=2a_1p_1q_1x+(p_1^2-q_1^2)y+\ln \theta_1, \quad E_2=E_1+2q_1(x+q_1y),\\
&E_3=8(a_1p_1q_1)^{2}(x^2+4p_1q_1(tx+y)y)+1,\\
&F_1=2(a_1p_1^2x+(tx+y)r_1^2),\quad F_2=2(b_1r_1^2x+(tx+y)p_1^2),\\
&F_3=p_1(p_1-r_1)(p_1^{-1}x+2y).
\end{align}
\end{subequations}
As an example, we illustrate the Jordan-block solution \eqref{lim-u(1,-1)} in Figure 3.

\begin{center}
\begin{picture}(120,100)
\put(-160,-23){\resizebox{!}{4cm}{\includegraphics{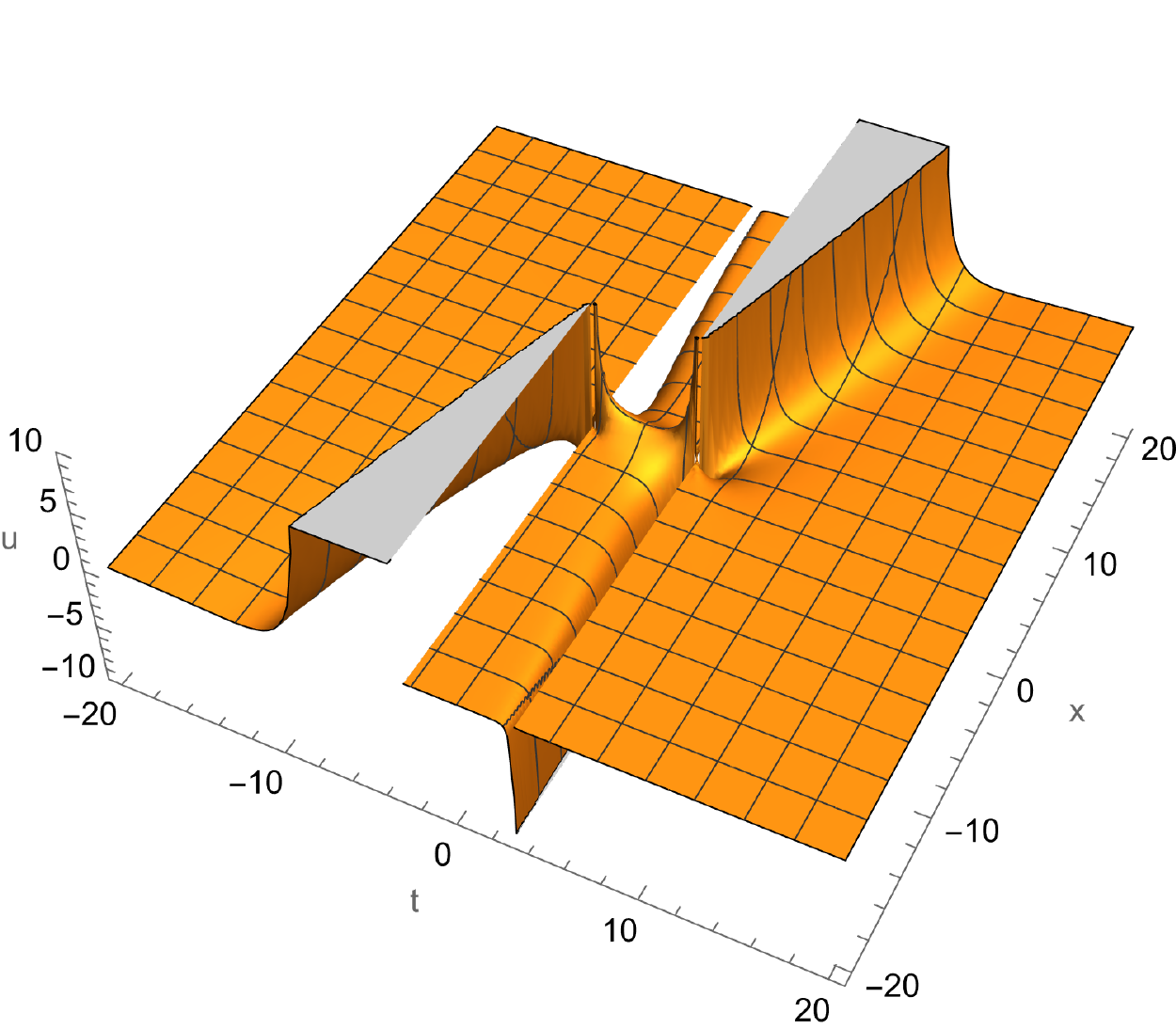}}}
\put(-10,-23){\resizebox{!}{4cm}{\includegraphics{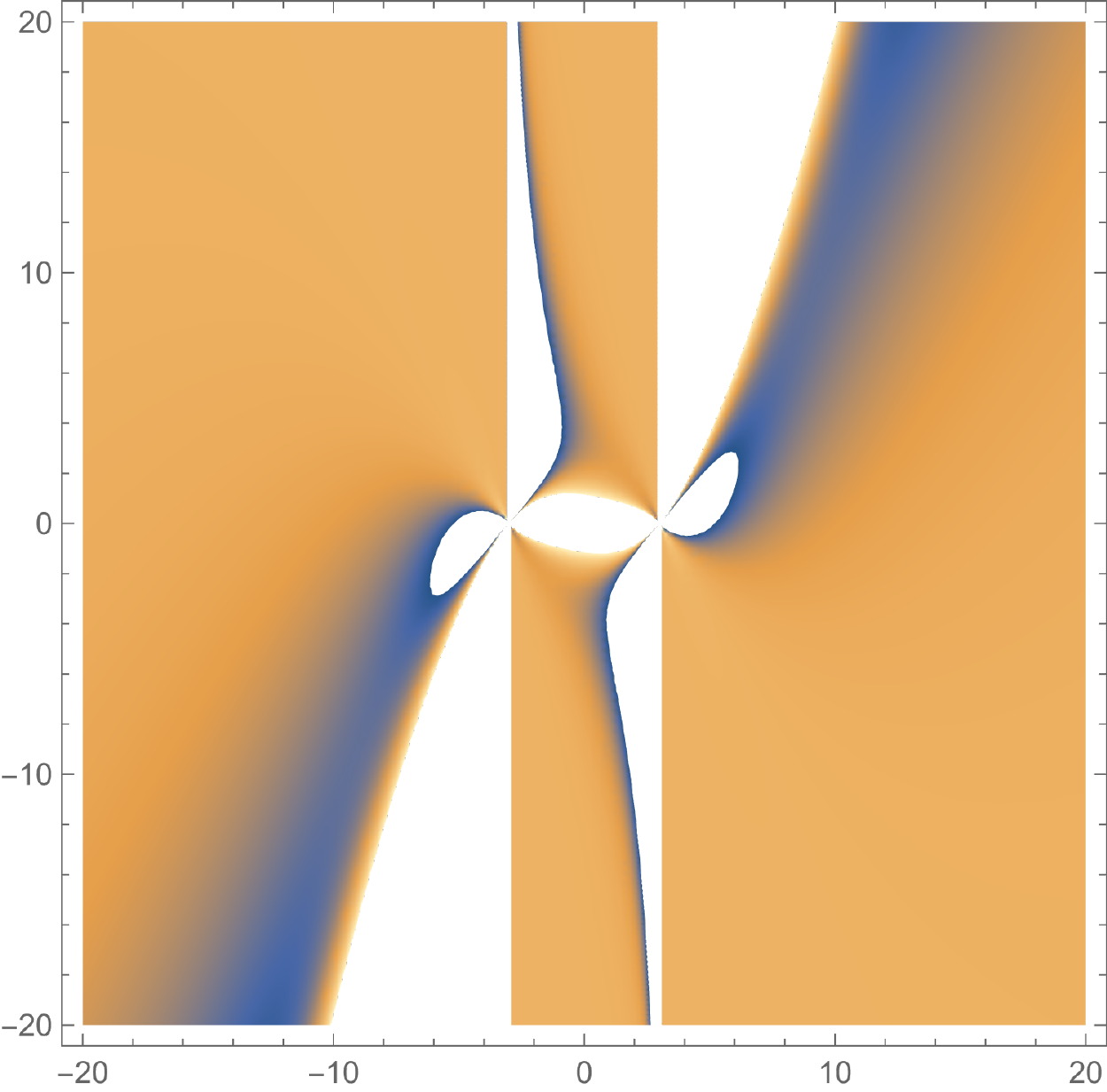}}}
\put(130,-23){\resizebox{!}{4cm}{\includegraphics{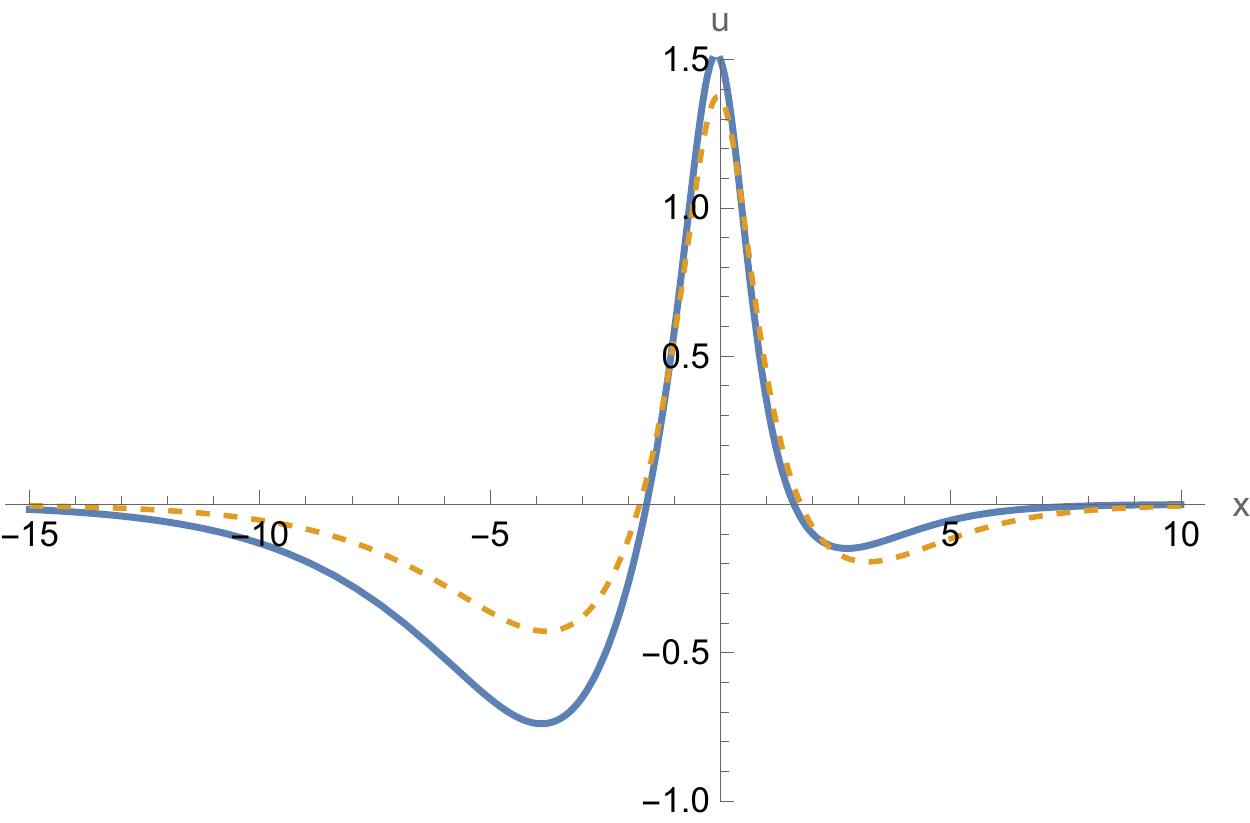}}}
\end{picture}
\end{center}
\vskip 20pt
\begin{center}
\begin{minipage}{16cm}{\footnotesize
\quad\qquad\qquad\qquad(a)\qquad\qquad\qquad \qquad\quad \quad \qquad\qquad  \quad (b) \qquad\qquad \qquad \qquad\qquad\qquad\qquad\qquad \quad (c)\\
{\bf Fig. 3} (a) shape and motion of $u$ given by \eqref{lim-u(1,-1)} for $a_1=3$, $\theta_1=1$ and $y=0$.
(b) a contour plot of (a) with range $x\in [-20, 20]$ and $t\in [-20,20]$.
(c) waves in solid and dotted line stand for  plot (a) at $t=1$ and $t=0.5$, respectively.}
\end{minipage}
\end{center}
	
\section{Complex nonlocal reduction of the system \eqref{n-(2+1)-DBS}}\label{sec-4}

In this section, we mainly discuss complex nonlocal reduction of the system \eqref{n-(2+1)-DBS}. The strategy
is similar to the real case, i.e., first reducing the system \eqref{n-(2+1)-DBS}, and second
obtaining exact solutions. Because of the complication brought by complex number, we just consider soliton solutions
for the resulting nonlocal equation.

We impose complex reduction
\begin{align}
\label{com-re1}
v(x,y,t)=\delta u^{*}(\sigma x,-y,-\sigma t),\quad \sigma,~~\delta=\pm 1,
\end{align}
and the system \eqref{n-(2+1)-DBS} reduces to the CNNBS-NLS equation
\begin{align}
\label{nn-BSE-C}
u_t= & -y(u_{xy}-2\delta u\partial^{-1}_{x}(uu^*(\sigma x,-y,-\sigma t))_y)-\frac{x}{2}(u_{xx}-2\delta u^2u^*(\sigma x,-y,-\sigma t)) \nn \\
& -u_{x}+\delta u\partial^{-1}_{x}uu^*(\sigma x,-y,-\sigma t).
\end{align}
Similarly, equation \eqref{nn-BSE-C} is preserved under transformation $u\rightarrow -u$ and
equation \eqref{nn-BSE-C} with $(\sigma,\delta)=(\pm 1, 1)$ and
with $(\sigma,\delta)=(\pm 1,-1)$ can be transformed into each other by taking $u\rightarrow \pm iu$.
When $\sigma=1$ \eqref{nn-BSE-C} is reverse-$(y,t)$ type and when $\sigma=-1$ \eqref{nn-BSE-C} is reverse-$(x,y)$ type.

We now investigate solutions of the equation \eqref{nn-BSE-C}, which can be
described by the following theorem.
\begin{Thm}
\label{Thm-BSEC-so}
Double Wronskian solutions of the CNNBS-NLS equation \eqref{nn-BSE-C} are given by $u=\frac{g}{f}$, in which
\begin{subequations}
\label{C-fg}
\begin{align}
f=& |\wh{\phi}^{(n)}(x,y,t);(-\sigma)^{n}T\wh{\phi}^{(n)^*}(\sigma x,-y,-\sigma t)|, \\
g=& 2|\wh{\phi}^{(n+1)}(x,y,t);(-\sigma)^{n}T\wh{\phi}^{(n-1)^*}(\sigma x,-y,-\sigma t)|,
\end{align}
\end{subequations}
where $\phi$ is the $2(n+1)$-th order column vectors defined by \eqref{phsi-solu1} and $2(n+1)$-th order matrices
\begin{align}
\label{com-kT-def}
K(t)=\left(\begin{array}{cc}
K_1(t) & \bm 0 \\
\bm 0 & K_2(t) \\
\end{array}\right),\quad T=\left(\begin{array}{cc}
T_1 & T_2 \\
T_3 & T_4 \\
\end{array}\right)
\end{align}
satisfy determining equations
\begin{align}
\label{nBSE-com-kT}
K(t)T+\sigma TK^*(-\sigma t)=0,\quad TT^*=\sigma\delta I,
\end{align}
and we require $\beta=T\alpha^*$.
\end{Thm}

In accordance with the equations \eqref{com-kT-def} and \eqref{nBSE-com-kT}, the solutions
of $K_i(t)$ and $T_j$ with $i=1,2$ and $j=1,2,3,4,$ are given in Table 2.
\begin{center}
\footnotesize \setlength{\tabcolsep}{8pt}
\renewcommand{\arraystretch}{1.5}
\begin{tabular}[htbp]{|l|l|l|l|}
\hline
$(\sigma,\delta)$& $K(t)$ & $T$\\
\hline
$(1,-1)$ &\eqref{com-kT-def} \mbox{with} $K_2(t)=-K^*_1(-t)$ & $T_1=T_4=\bm 0, \quad T_3=-T_2=I$ \\
\hline
$(1,1)$ & \eqref{com-kT-def} \mbox{with} $K_2(t)=-K^*_1(-t)$ & $T_1=T_4=\bm 0, \quad T_3=T_2=I$ \\
\hline
$(-1,-1)$ &\eqref{com-kT-def} \mbox{with} $K_2(t)=K^*_1(t)$& $T_1=T_4=\bm  0, \quad T_3=-T_2=I$ \\
\hline
$(-1,1)$ & \eqref{com-kT-def} \mbox{with} $K_2(t)=K^*_1(t)$ & $T_1=T_4=\bm 0, \quad T_3=T_2=I$ \\
\hline
\end{tabular}
\end{center}
\begin{center}
\begin{minipage}{7cm}{\footnotesize
{Table 2. \emph{$K(t)$ and $T$ for equation \eqref{nBSE-com-kT}.}}}
\end{minipage}
\end{center}
	
To derive the soliton solutions, we set
\begin{align}
K_1=\mbox{Diag}(k_1,~k_2,~\ldots,~k_{n+1}),
\end{align}
in which $k_j=(t-c_j)^{-1}$, where $c_{j},~j=1,2,\ldots,n+1$ are complex constants.
Besides, we employ notation $l_j=(t+c_j^*)^{-1}$.
With different $(\sigma,\delta)$, the one-soliton solution of equation \eqref{nn-BSE-C} can be described as
\begin{subequations}
\label{cn-noni-(2+1)-DBSE-1soli}
\begin{align}
&\label{cn-noni-(2+1)-DBSE-u(1,-1)}
u_{\sigma=1,\delta=-1}=-\frac{2\alpha_1{\wt{\alpha}}_{1}(k_1-l_1)}
                       {|\alpha_1|^2e^{(k^*_1-l_1)x-({k_1^*}^2+l_1^2)y}+|{\wt{\alpha}}_{1}|^2e^{(l^*_1-k_1)x-({l_1^*}^2+k_1^2)y}},\\
&\label{cn-noni-(2+1)-DBSE-u(1,1)}
u_{\sigma=1,\delta=1}=-\frac{2\alpha_1{\wt{\alpha}}_{1}(k_1-l_1)}
                      {|\alpha_1|^2e^{(k^*_1-l_1)x-({k_1^*}^2+l_1^2)y}-|{\wt{\alpha}}_{1}|^2e^{(l^*_1-k_1)x-({l_1^*}^2+k_1^2)y}},\\
&\label{cn-noni-(2+1)-DBSE-u(-1,-1)}
u_{\sigma=-1,\delta=-1}=-\frac{2\alpha_1{\wt{\alpha}}_{1}(k_1-k_1^*)}
			           {|\alpha_1|^2e^{-2k^*_1x-2{k_1^*}^2y}-|{\wt{\alpha}}_{1}|^2e^{-2k_1x-2k_1^2y}},\\
&\label{cn-noni-(2+1)-DBSE-u(-1,1)}	
u_{\sigma=-1,\delta=1}=-\frac{2\alpha_1{\wt{\alpha}}_{1}(k_1-k_1^*)}
                       {|\alpha_1|^2e^{-2k^*_1x-2{k_1^*}^2y}+|{\wt{\alpha}}_{1}|^2e^{-2k_1x-2k_1^2y}},
\end{align}
\end{subequations}
with module $|\cdot|$.
	
Now let us describe the dynamic behaviors and characteristics of $|u|^2$ given by \eqref{cn-noni-(2+1)-DBSE-1soli}.
Above all, we concentrate on solutions \eqref{cn-noni-(2+1)-DBSE-u(1,-1)} and \eqref{cn-noni-(2+1)-DBSE-u(1,1)}.
Denoting $c_1=\mu+i\nu$ and substituting it back into \eqref{cn-noni-(2+1)-DBSE-u(1,-1)} and \eqref{cn-noni-(2+1)-DBSE-u(1,1)} give rise to
\begin{subequations}
\label{cn-noni-(2+1)-DBSE-mude}	
\begin{align}
& \label{cn-noni-(2+1)-DBSE-|u(1,-1)|}
|u|^2_{\sigma=1,\delta=-1}=\frac{16|\alpha_1{\wt{\alpha}}_{1}|^2\mu ^2e^{2X_1+Y_1}}
      {(4t^2\nu^2+\xi^2)(|\alpha_1|^4e^{4X_1}+|{\wt{\alpha}}_{1}|^4+2|\alpha_1{\wt{\alpha}}_{1}|^2e^{2X_1}\cos(2Y_2))},\\
&\label{cn-noni-(2+1)-DBSE-|u(1,1)|}
      |u|^2_{\sigma=1,\delta=1}=\frac{16|\alpha_1{\wt{\alpha}}_{1}|^2\mu ^2e^{2X_1+Y_1}}
      {(4t^2\nu^2+\xi^2)(|\alpha_1|^4e^{4X_1}+|{\wt{\alpha}}_{1}|^4-2|\alpha_1{\wt{\alpha}}_{1}|^2e^{2X_1}\cos(2Y_2))},
\end{align}
\end{subequations}
in which
\begin{subequations}
\begin{align*}
& \xi=t^2-\mu^2-\nu^2,\quad \xi_1=t^2-\mu^2+\nu^2,\quad \xi_2=t^2+\mu^2-\nu^2,\\ \nn
& X_1=2\mu x\frac{\xi_1-2\nu^2}{\xi_1^2+4\mu^2\nu^2},\quad Y_1=2y\frac{\xi_1^2\xi_2-4\mu^2\nu^2(2\xi_1+\xi_2)}{(\xi_1^2+4\mu^2\nu^2)^2},
\quad Y_2=4\mu\nu y\frac{(\xi_1^2+2\xi_1\xi_2-4\mu^2\nu^2)}{(\xi_1^2+4\mu^2\nu^2)^2}.
\end{align*}
\end{subequations}
For solution \eqref{cn-noni-(2+1)-DBSE-|u(1,-1)|},
when $y=0$, it is a nonsingular wave while when $y\neq 0$ it has singularities along
\begin{align}
\label{sigu}
y(t)=\frac{\kappa \pi(\xi_1^2+4\mu^2\nu^2)^2}{4\mu\nu(\xi_1^2+2\xi_1\xi_2-4\mu^2\nu^2)}, \quad \kappa\in\mathbb{Z}.
\end{align}
For solution \eqref{cn-noni-(2+1)-DBSE-|u(1,1)|}, it always has singularities along \eqref{sigu} regardless of $y=0$ or not.
We depict these two solutions in Fig. 4.

\begin{center}
\begin{picture}(120,100)
\put(-120,-23){\resizebox{!}{4cm}{\includegraphics{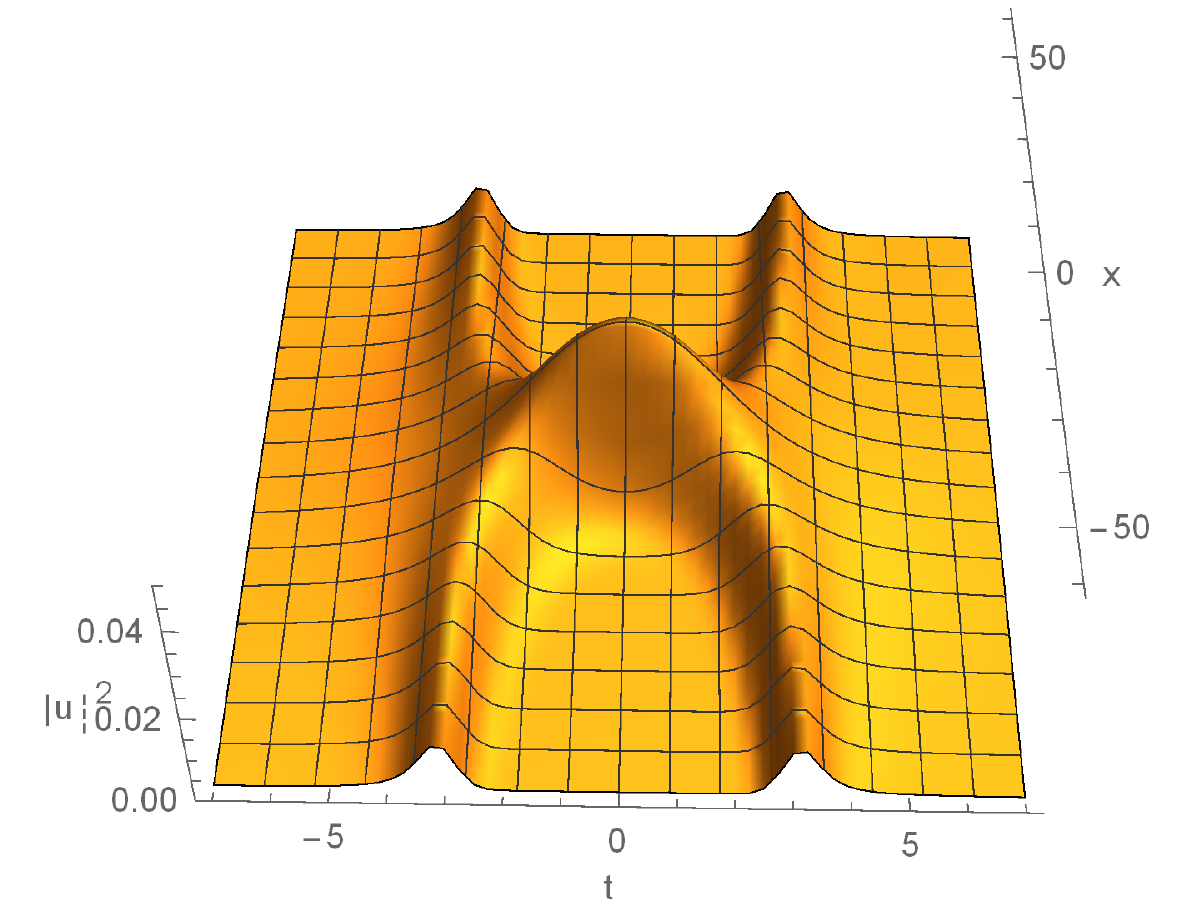}}}
\put(100,-23){\resizebox{!}{4cm}{\includegraphics{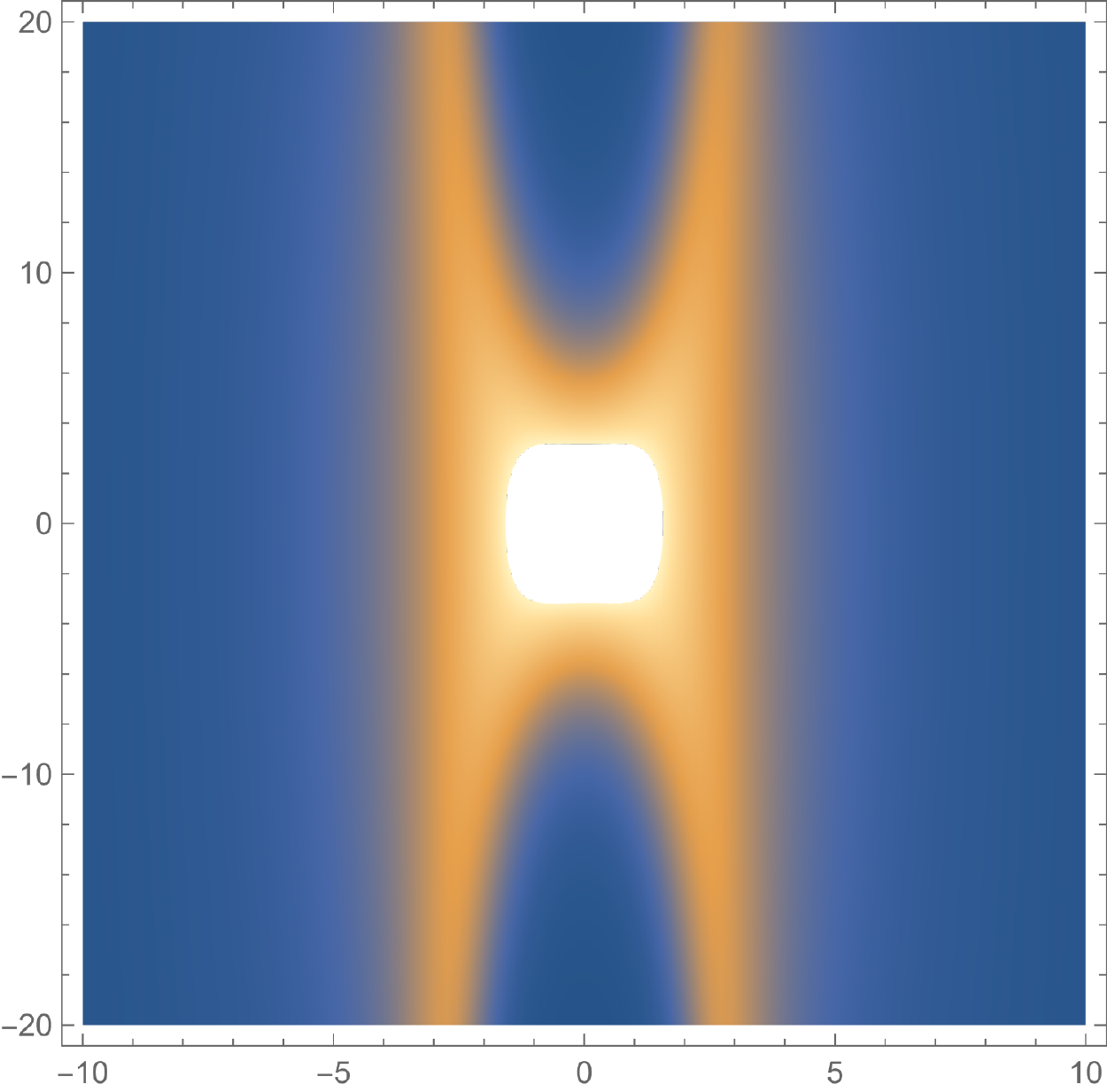}}}
\end{picture}
\end{center}
\vskip 20pt
\begin{center}
\begin{minipage}{15cm}{\footnotesize
\quad\qquad\qquad\qquad\qquad\qquad\quad(a)\qquad\qquad\qquad\qquad\qquad\qquad\qquad\qquad\quad \qquad \quad (b)}
\end{minipage}
\end{center}
\vskip 10pt
\begin{center}
\begin{picture}(120,80)
\put(-120,-23){\resizebox{!}{4cm}{\includegraphics{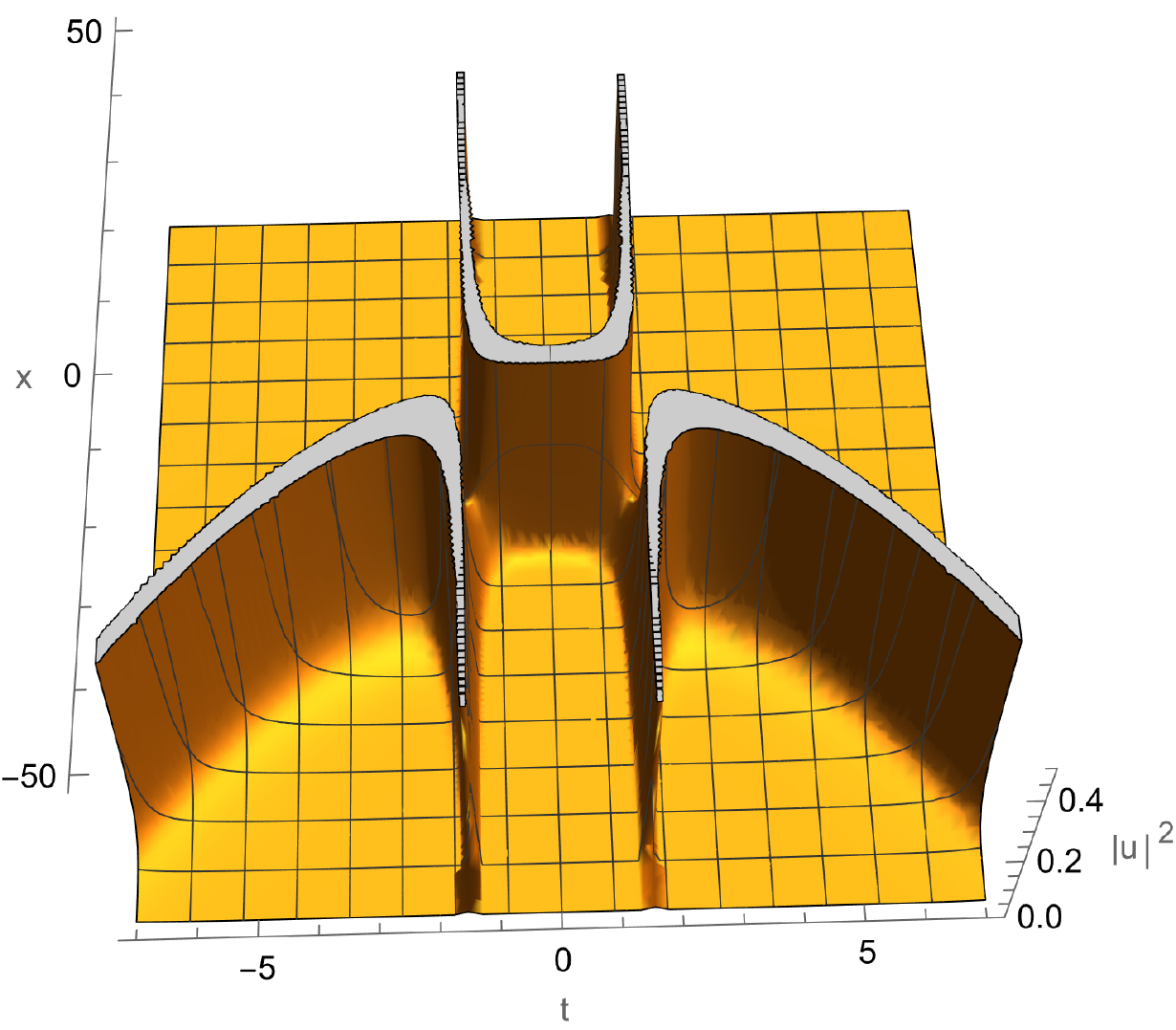}}}
\put(100,-23){\resizebox{!}{4cm}{\includegraphics{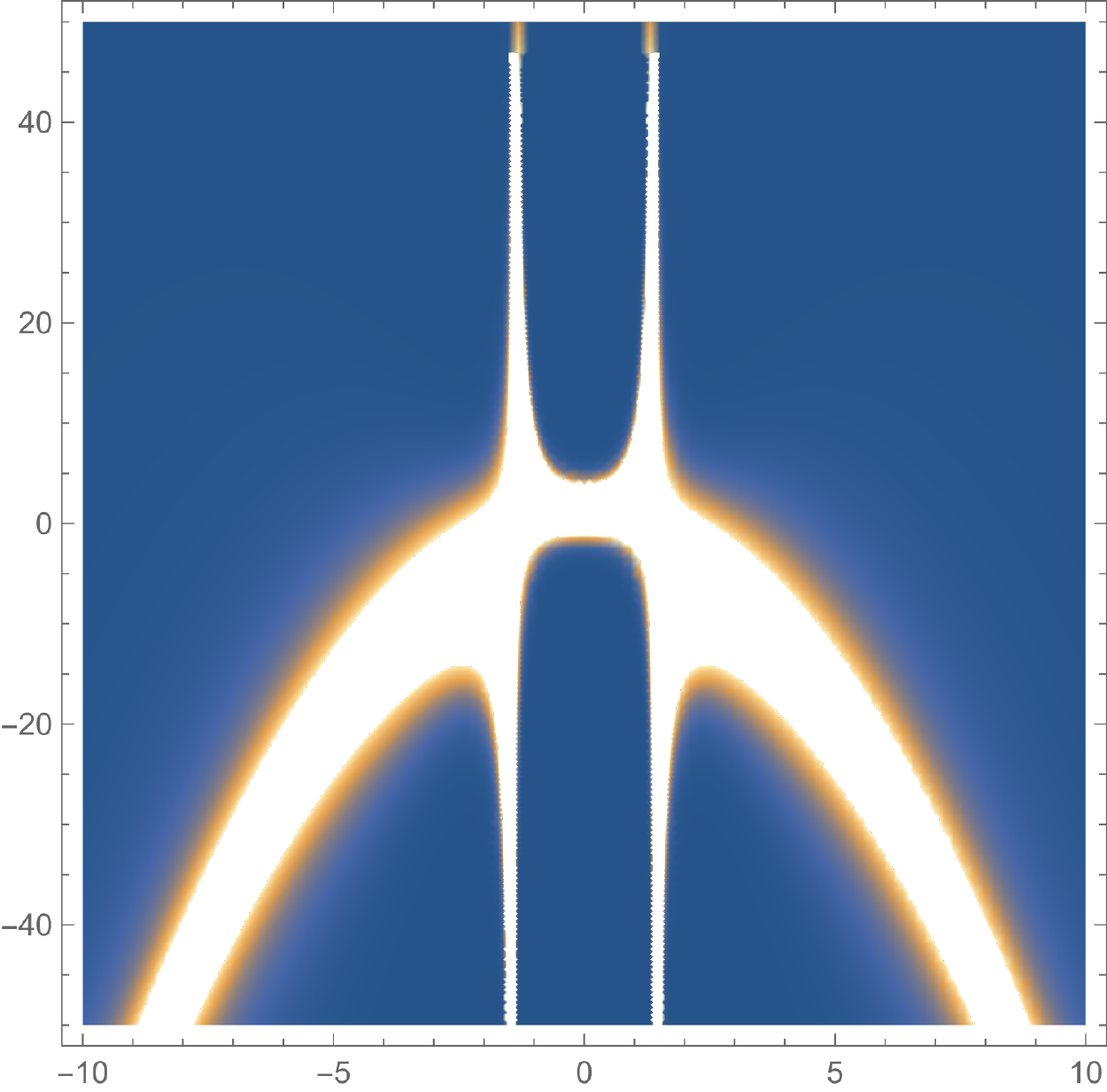}}}
\end{picture}
\end{center}
\vskip 20pt
\begin{center}
\begin{minipage}{15cm}{\footnotesize
\quad\qquad\qquad\qquad\qquad\quad(c)\qquad\qquad\qquad\qquad\qquad\qquad\qquad\qquad\qquad\quad \qquad \quad (d)\\
{\bf Fig. 4} (a) shape and motion of \eqref{cn-noni-(2+1)-DBSE-|u(1,-1)|} in which $c_1=1+3i,~\alpha_1=\wt{\alpha}_1=1+i$ and $y=0$.
(b) density plot of (a) with range $x\in[-20;20]$ and $t\in[-10;10]$.
(c) shape and motion of \eqref{cn-noni-(2+1)-DBSE-|u(1,1)|} in which $c_1=1+i,~\alpha_1=4,~\wt{\alpha}_{1}=1$ and $y=0$.
(d) density plot of (c) with range $x\in[-50;50]$ and $t\in[-10;10]$.}
\end{minipage}
\end{center}

To proceed, we consider solutions \eqref{cn-noni-(2+1)-DBSE-u(-1,-1)} and \eqref{cn-noni-(2+1)-DBSE-u(-1,1)}. Under the expansion $c_1=\mu+i\nu$, we observe that	
\begin{subequations}
\begin{align}
&\label{cn-noni-(2+1)-DBSE-|u(-1,-1)|}
|u|^2_{\sigma=-1,\delta=-1}=\frac{16|\alpha_1{\wt{\alpha}}_{1}|^2\nu ^2e^{2Z_1}}
{((t-\mu)^2+\nu^2)^2(|\alpha_1|^4+|{\wt{\alpha}}_{1}|^4-2|\alpha_1\wt{\alpha}_{1}|^2\cos(2Z_2))},\\
&\label{cn-noni-(2+1)-DBSE-|u(-1,1)|}	
|u|^2_{\sigma=-1,\delta=1}=\frac{16|\alpha_1{\wt{\alpha}}_{1}|^2\nu ^2e^{2Z_1}}
{((t-\mu)^2+\nu^2)^2(|\alpha_1|^4+|{\wt{\alpha}}_{1}|^4+2|\alpha_1\wt{\alpha}_{1}|^2\cos(2Z_2))},
\end{align}
\end{subequations}
where
\begin{align*}
Z_1=\frac{(t-\mu)((t-\mu)^2+\nu^2)x+2((t-\mu)^2-\nu^2)y}{((t-\mu)^2+\nu^2)^2}, ~~
Z_2=\frac{((t-\mu)^2+\nu^2)x+2\nu(2(t-\mu)y)}{((t-\mu)^2+\nu^2)^2}.
\end{align*}
\eqref{cn-noni-(2+1)-DBSE-|u(-1,1)|} has analogous property
with \eqref{cn-noni-(2+1)-DBSE-|u(-1,-1)|}. Solution \eqref{cn-noni-(2+1)-DBSE-|u(-1,-1)|} is nonsingular when $|\alpha_1| \neq |{\wt{\alpha}}_{1}|$.
While when $|\alpha_1|=|{\wt{\alpha}}_{1}|$, solution \eqref{cn-noni-(2+1)-DBSE-|u(-1,-1)|}
has singularities along
\begin{align}
x(t)=\frac{\kappa \pi((t-\mu)^2+\nu^2)}{2\nu}, \quad \kappa\in\mathbb{Z}.
\end{align}
Because of the involvement of cosine function in denominator, there is quasi-periodic phenomenon.
We depict solution \eqref{cn-noni-(2+1)-DBSE-|u(-1,-1)|} in Fig. 5.

\begin{center}
\begin{picture}(120,100)
\put(-120,-23){\resizebox{!}{4cm}{\includegraphics{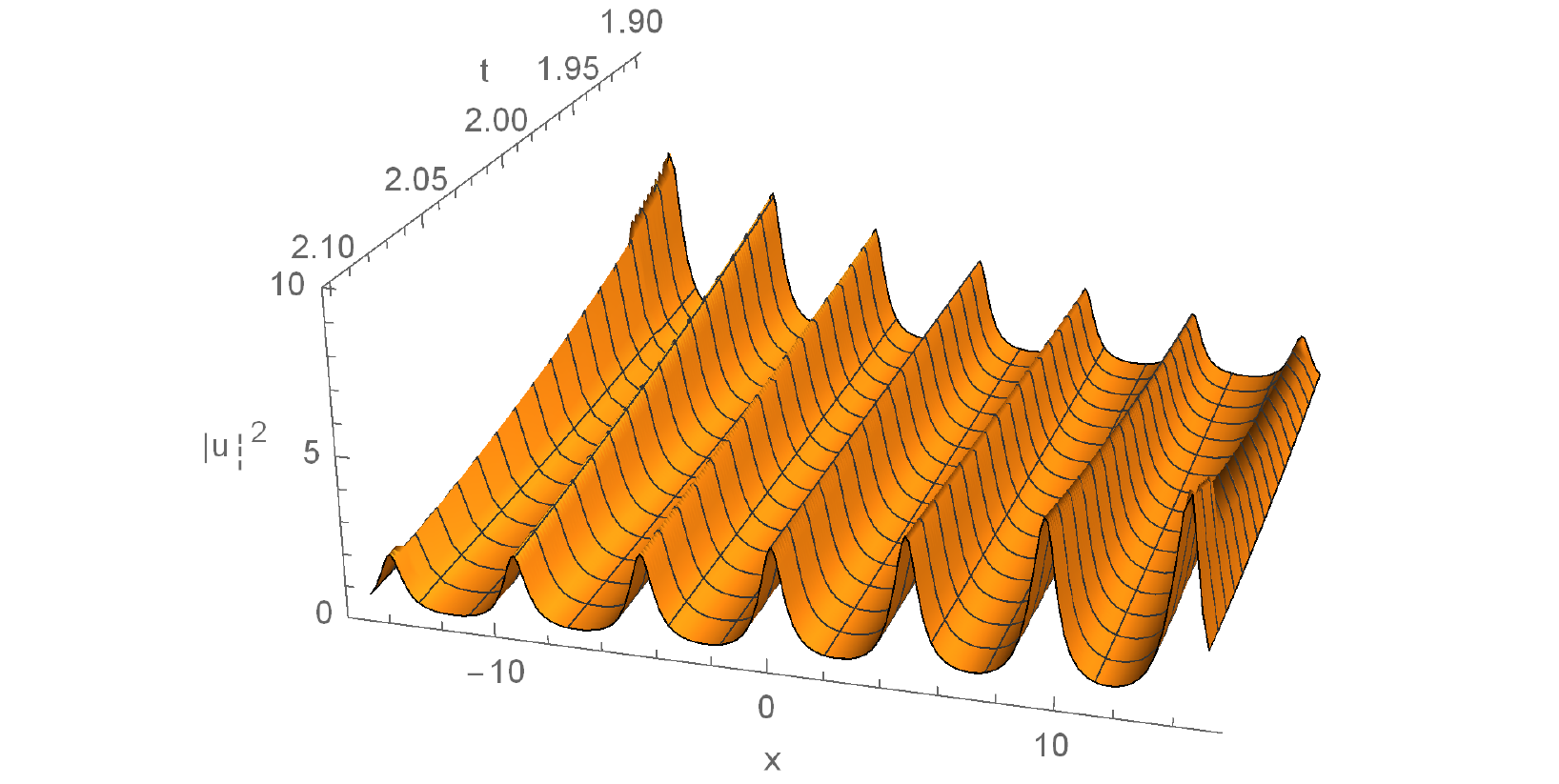}}}
\put(100,-23){\resizebox{!}{4cm}{\includegraphics{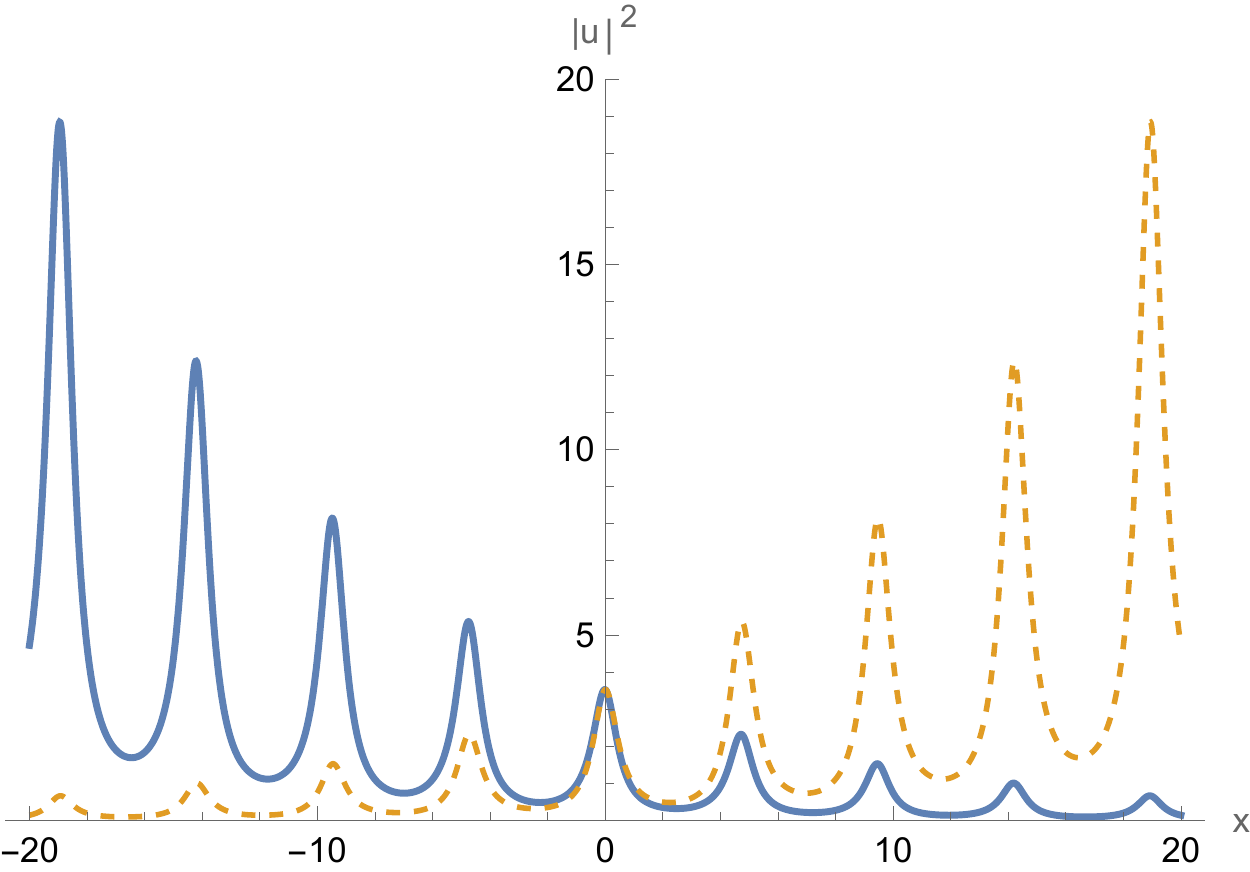}}}
\end{picture}
\end{center}
\vskip 20pt
\begin{center}
\begin{minipage}{15cm}{\footnotesize
\quad\qquad\qquad\qquad\qquad\qquad\quad\qquad(a)\qquad\qquad\qquad\qquad\qquad\qquad\qquad\qquad\quad \qquad \quad (b)}
\end{minipage}
\end{center}
\vskip 10pt
\begin{center}
\begin{picture}(120,80)
\put(-120,-23){\resizebox{!}{4cm}{\includegraphics{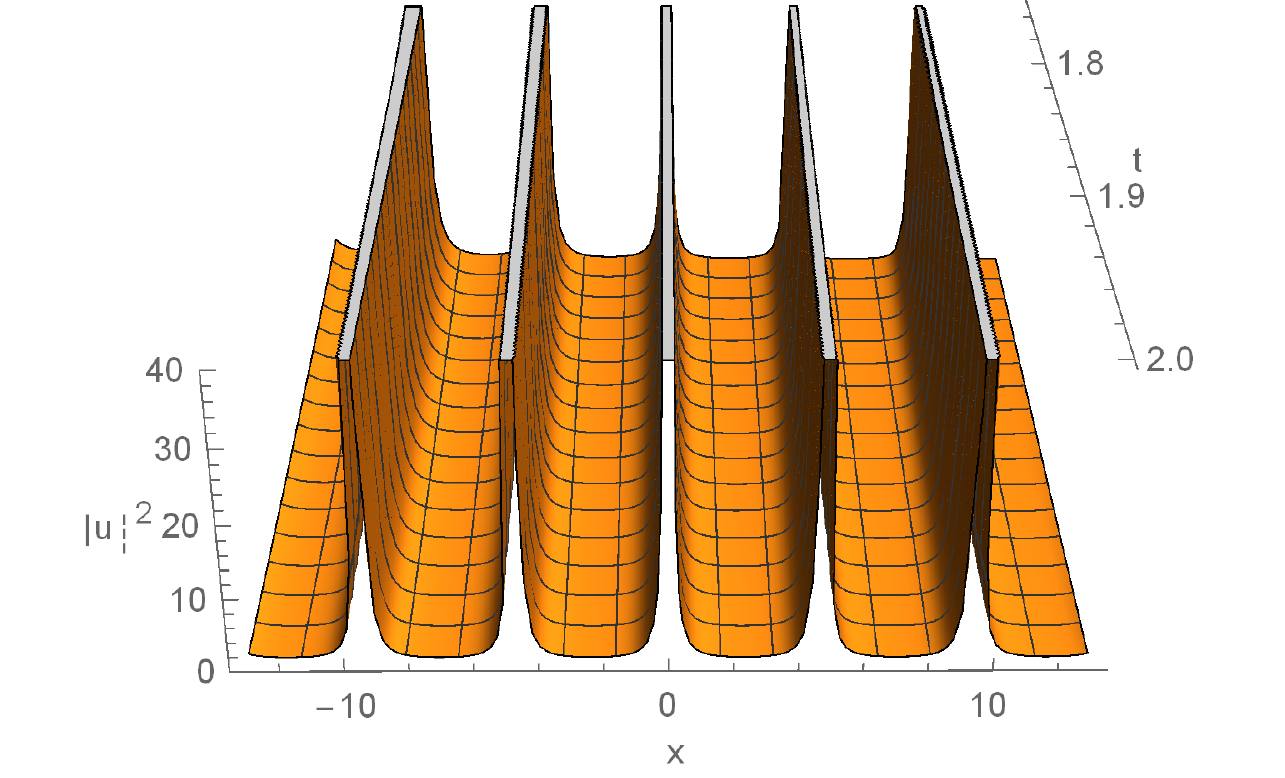}}}
\put(100,-23){\resizebox{!}{4cm}{\includegraphics{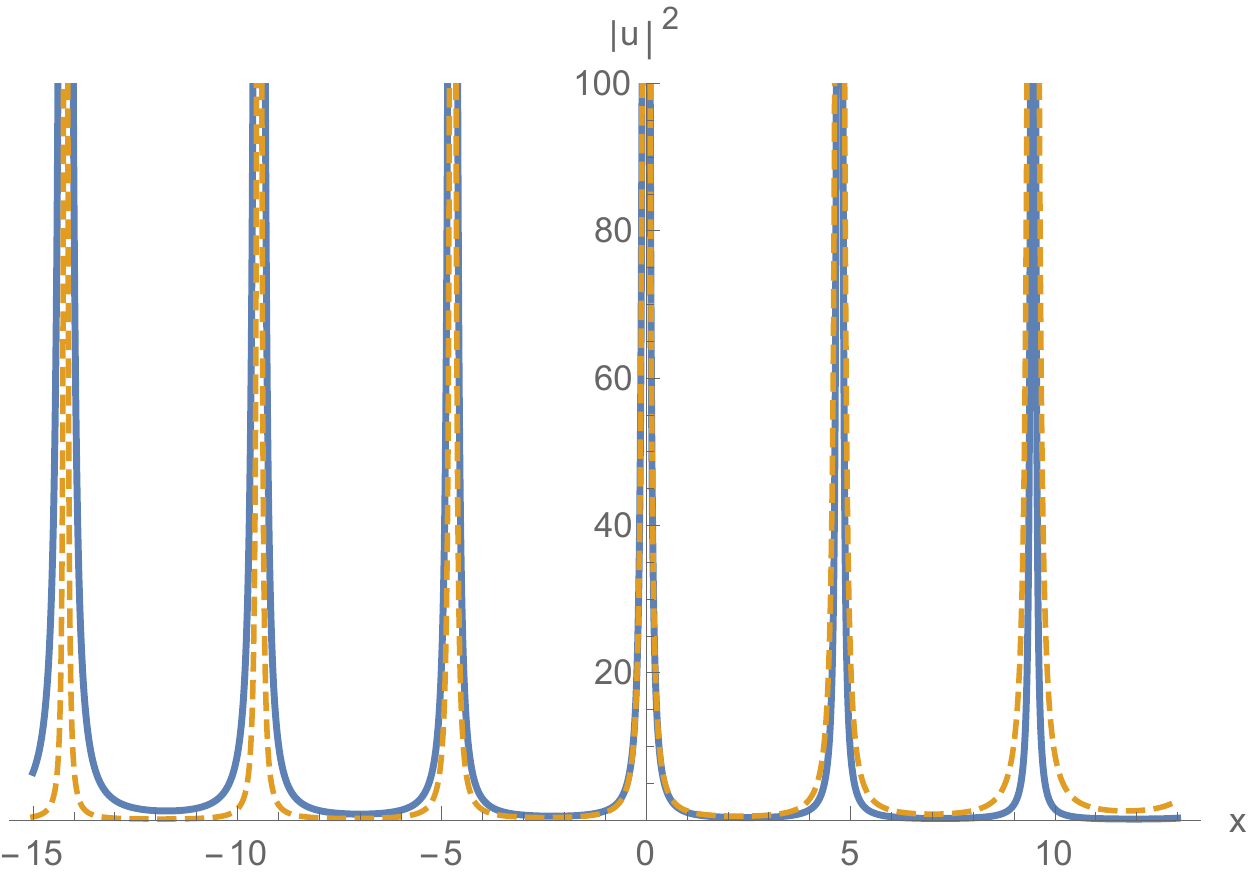}}}
\end{picture}
\end{center}
\vskip 20pt
\begin{center}
\begin{minipage}{15cm}{\footnotesize
\quad\qquad\qquad\qquad\qquad\qquad\quad\qquad(c)\qquad\qquad\qquad\qquad\qquad\qquad\qquad\qquad\quad \qquad \quad (d)\\
{\bf Fig. 5} (a) shape and motion of \eqref{cn-noni-(2+1)-DBSE-|u(-1,-1)|} in which $c_1=2+3i,~\alpha_1=1,~\wt{\alpha}_{1}=1+i,~y=0$.
(b) waves in solid and dotted line stand for plot (a) at $t=1.8$ and $t=2.2$, respectively.
(c) shape and motion of \eqref{cn-noni-(2+1)-DBSE-|u(-1,-1)|} in which $c_1=2+3i,~\alpha_1=\wt{\alpha}_{1}=1+i,~y=0$.
(d) waves in solid and dotted line stand for plot (c) at $t=1.8$ and $t=2.2$, respectively.}
\end{minipage}
\end{center}

\section{Conclusions}\label{sec-5}

In this work, we use the double Wronskian reduction technique
to consider real and complex nonlocal reductions of the NBS-AKNS \eqref{n-(2+1)-DBS}. Soliton solutions and
Jordan-block solutions for the resulting nonlocal equations are obtained by solving the determining
equations. Dynamics of one-soliton, two-solitons and the simplest Jordan-block solutions are analyzed and illustrated.
We find that nonlocal type for the real case and complex case are the same. Both of these two cases
have reverse-$(y,t)$ type and reverse-$(x,y)$ type nonlocal reductions. The reverse-$(y,t)$ type nonlocal reduction
is also admitted by the isospectral (2+1)-dimensional breaking soliton AKNS equation.
While the reverse-$(x,y)$ type nonlocal reduction is just admitted by the NBS-AKNS \eqref{n-(2+1)-DBS}.
In recent paper \cite{AbMu-PLA}, several novel integrable nonlocal systems,
including the shifted PT symmetric and the shifted time delay nonlocal NLS equations, were proposed and
have attracted more and more attention \cite{Gurses,Zhang-shift}. How to
extend the results in the present paper to the shifted space-time nonlocal isospectral and
nonisospectral (2+1)-dimensional breaking soliton equations
is an interesting questions worth consideration. We hope that the results given in the present paper can be useful to study the nonlocal integrable system,
specially to the nonlocal nonisospectral (2+1)-dimensional integrable systems.
	

\vspace{3mm}
\noindent {\bf Acknowledgments}: This project is supported by the National Natural Science Foundation of China (Nos. 12071432, 11401529)
and the Natural Science Foundation of Zhejiang Province (No. LY18A010033).
	


\end{document}